\documentclass[letterpaper]{article} 
\usepackage{aaai24}  
\usepackage{times}  
\usepackage{helvet}  
\usepackage{courier}  
\usepackage[hyphens]{url}  
\usepackage{graphicx} 
\def\UrlFont{\rm}  
\usepackage{natbib}  
\usepackage{caption} 
\usepackage{algorithm}
\usepackage{algorithmic}

\usepackage{newfloat}
\usepackage{listings}
\DeclareCaptionStyle{ruled}{labelfont=normalfont,labelsep=colon,strut=off} 
\floatstyle{ruled}
\newfloat{listing}{tb}{lst}{}
\floatname{listing}{Listing}
\title{Generalized Bradley-Terry Models for Score Estimation from Paired Comparisons}
\author{
    Julien Fageot\textsuperscript{\rm 1,}\textsuperscript{\rm 2}, 
    Sadegh Farhadkhani\textsuperscript{\rm 2}, 
    Lê-Nguyên Hoang\textsuperscript{\rm 1,}\textsuperscript{\rm 3}, and
    Oscar Villemaud\textsuperscript{\rm 1,}\textsuperscript{\rm 2}
}
\affiliations{
    \textsuperscript{\rm 1} Tournesol Association \\
    \textsuperscript{\rm 2} EPFL \\
    \textsuperscript{\rm 3} Calicarpa \\
    julien.fageot@gmail.com, sadegh.farhadkhani@epfl.ch, len@tournesol.app, oscar.villemaud@epfl.ch

}

\usepackage{bibentry}

\usepackage[utf8]{inputenc}
\usepackage{graphicx}
\usepackage{xcolor}
\usepackage{mathtools}
\usepackage{amsmath}
\usepackage{amssymb}
\usepackage{amsthm}
\usepackage{xfrac}
\usepackage{bm}
\usepackage{enumitem}
\usepackage{ifthen}
\usepackage[font=small]{subfig}
\usepackage{wasysym}
\usepackage{float}

\usepackage{todonotes}

\newcommand{\norm}[2]{\left\lVert{#1}\right\rVert_{#2}}

\DeclareMathOperator*{\argmin}{arg\,min}

\newtheorem{definition}{Definition}
\newtheorem{proposition}{Proposition}
\newtheorem{theorem}{Theorem}
\newtheorem{lemma}{Lemma}

\newcommand{\Z}{\mathbb Z}
\newcommand{\setR}{\mathbb R}

\newcommand{\comparison}{r}

\newcommand{\R}{\mathbb{R}}

\newcommand{\Comp}{\mathrm{R}}

\newcommand{\Dc}{\textsc{Supp}(f)}
\newcommand{\Ds}{\mathbb{D}_{\mathrm{s}}}

\begin{document}

\maketitle

\begin{abstract}
Many applications, e.g. in content recommendation, sports, or recruitment, 
leverage the comparisons of alternatives to score those alternatives.
The classical Bradley-Terry model and its variants have been widely used to do so. 
The historical model considers binary comparisons (victory/defeat) between alternatives, while more recent developments allow finer comparisons to be taken into account.
In this article, we introduce a probabilistic model encompassing a broad variety of paired comparisons that can take discrete or continuous values. We do so by considering a well-behaved subset of the exponential family,
which we call the family of \emph{generalized Bradley-Terry (GBT) models},
as it includes the classical Bradley-Terry model and many of its variants.
Remarkably, we prove that all GBT models are guaranteed to yield a strictly convex negative log-likelihood.
Moreover, assuming a Gaussian prior on alternatives' scores,
we prove that the maximum a posteriori (MAP) of GBT models, 
whose existence, uniqueness and fast computation are thus guaranteed,
varies monotonically with respect to comparisons 
(the more A beats B, the better the score of A)
and is Lipschitz-resilient with respect to each new comparison
(a single new comparison can only have a bounded effect on all the estimated scores).
These desirable properties make GBT models appealing for practical use. We illustrate some features of GBT models on simulations.

\end{abstract}


\section{Introduction}
\label{sec:intro}

In many settings, alternatives are rather compared than individually scored.
Typically, in chess, football, tennis, judo, or cycling, individuals and teams compete against one another.
Similarly, students and job candidates are arguably easier to compare,
rather than to assess directly.
In fact, comparative judgments are implicitly performed all the time in online applications,
as users often have to select content, applications, or products to consume or purchase, 
within a set of proposed alternatives.
However, ranking all (or a subset of top) alternatives is often demanded.
Many sport competitions identify a current number-one player or team, 
job candidates are ordered for hiring procedures
and recommendation AIs must select a handful of content to recommend.
Such rankings are often produced based on scoring inferred from comparisons.
Scores also allow to reflect the fact that an alternative vastly outperforms a particularly bad alternative,
while it is only slightly better than a third alternative.

Transforming comparisons into scores is not straightforward, especially when the comparisons are noisy.
Typically, better sport teams are sometimes unfortunate, and end up losing against less competitive teams, which is sometimes called the ``beauty of sport".
In other applications, the comparative judgments may vary because they are made by different individuals,
or simply because it is hard for humans to remain consistent in their sequential judgments.

\subsection{Contributions}
In this paper, we introduce and analyze a natural and well-behaved family of probabilistic models 
that convert observed comparisons into individual scoring.
Essentially, our family of models, which we call the generalized Bradley-Terry (GBT) family, 
is obtained by considering a subset of the exponential family of particular interest.
Interestingly, a practitioner then merely needs to define 
how they expect two equally good alternatives to be compared,
to effortlessly construct a unique model of our family.
We show that the GBT family generalizes the well-known and widely studied Bradley-Terry model~\cite{BradleyTerry52} 
(which is limited to binary comparisons in its historical version),
as well as other more recent models~\cite{guo2012score,kristof2019user}.
In fact, we highlight the generality of our GBT models by briefly analyzing multiple noteworthy instances,
including the $K$-nary-GBT, the Gaussian-GBT, the Uniform-GBT, and the Poisson-GBT.

Remarkably, as our key contribution, we prove that, 
given a Gaussian Bayesian prior on alternatives' scores, 
all GBT models are guaranteed to yield several desirable properties. 
Namely, first, all yield a \emph{strongly convex negative log-posterior}, 
which means that the maximum-a-posteriori (MAP) can be quickly computed 
by any optimizer for strongly convex losses.
Second, we prove that the MAP scores vary \emph{monotonically} with the comparisons.
This guarantees that alternatives will always be incentivized to ``win comparisons''
by as large of a margin as possible.
Finally, we prove that the MAP scores are \emph{Lipschitz-resilient} to each new comparison,
thereby guaranteeing that any outlier comparison will only have a limited impact on the estimated scores.
This is especially important in the context of online applications, where misclicks are extremely common.

These properties of the GBT models make them appealing to practitioners.
In fact, they have been deployed on the online collaborative Tournesol platform~\cite{tournesol},
which aims to construct a secure, ethical and collaboratively governed content recommendation algorithm.
To do so, the Tournesol platform asks its contributors 
to provide comparative judgments of which of a pair of videos should be more often recommended 
by the Tournesol recommender.
They reportedly use a GBT model, among other tools, to transform such comparisons into video scores.
An openly available implementation of GBT, under AGPL license, 
was made available in the python package called \texttt{solidago}.

\subsection{Related Works}

An important motivation behind this work is to lay the theoretical foundations of the method of estimating personal scores from paired comparisons made by Tournesol users~\cite{tournesol}. We therefore take inspiration from similar methods applied in contexts where they have proven their worth. One classical approach consists of defining a probabilistic model that conceives the comparisons as random variations from intrinsic scores (modeled as latent variables) of these comparisons~\cite{david1963method}.

The Bradley-Terry model~\cite{BradleyTerry52}, which follows the pioneering work of Zermelo~\cite{zermelo1929berechnung}, proposes to consider this task from binary comparisons (victory or defeat).
These ideas are the root of Elo rating system used in chess~\cite{elo1978rating} and other competitive contexts\footnote{For instance, the Elo rating of tennis players are computed here: \url{https://tennisabstract.com/reports/wta_elo_ratings.html}.}. 
A similar approach was followed independently by Thurstone~\cite{thurstone1927law} in psychophysics.
The Bradley-Terry and Thurstonian models have been generalized in many ways, for instance in order to include ties~\cite{davidson1970extending,rao1967ties}. Refinements of the Elo rating such as Glicko~\cite{glickman1999parameter} or TrueSkill~\cite{herbrich2006trueskill} have been introduced and practically used, together with online or Bayesian techniques to compute score estimators~\cite{hunter2004mm,cattelan2012models}.

Recent approaches considered different comparison models, including unbounded Poisson~\cite{maher1982modelling} and  Skellam (i.e. symmetric Poisson) models~\cite{karlis2009bayesian}, or continuous-domain Gaussian models~\cite{guo2012score,maystre2019pairwise,kristof2019user}. 
To the best of our knowledge, there is no theory covering all of these models which are based on different modeling of the comparisons. This article proposes precisely to introduce a unified framework including them all.

\subsection{Outline}
The rest of the paper is organized as follows.
First, we define the setting, and introduce the GBT models.
We especially stress the importance of the \emph{cumulant-generating functions},
wherein lie so many of the well-behaved properties of GBT models.
We then introduce MAP estimators based on GBT models, given a Gaussian prior on scores,
and highlight their basic computational and statistical properties.
Next, we define the \emph{monotonicity} of score estimators 
and prove that any GBT MAP estimator has this desirable property. 
In the following section,  we (re)define \emph{Lipchitz-resilience} to user's modifications,
and show that GBT MAP estimators with bounded scores are Lipschitz-resilient. 
We then exemplify GBT models, from the historical binary Bradley-Terry model to continuous-domain ones, provide illustrative simulations, and finally conclude in the last section. 

\section{Generalized Bradley-Terry Models} \label{sec:GBT}

In this section, we introduce the setting and the GBT models.
We then redefine cumulant-generating functions and highlight a remarkable result about these functions,
which will prove, as a corollary, that the MAP of GBT models can be efficiently computed.
Finally, we draw the connection with the historical Bradley-Terry model.

\subsection{The Setting}

Consider a set $\mathcal{A}= \{a\}_{a \in \mathcal{A}}$ of alternatives with cardinal $A = |\mathcal{A}|$. 
We assume that comparisons $r_{ab}$ between alternatives $a$ and $b$ have been made, 
for a (potentially small) subset of pairs of alternatives.
A positive comparison $r_{ab} > 0$ means that $b$ is judged better than $a$, 
and increases when the judgment is more pronounced.
We assume that $r_{ab} = - r_{ba}$ (i.e. $a$ beats $b$ if and only if $b$ is beaten by $a$).

We denote by 
$\mathcal{C}$ the set of pairs $(ab) \in \mathcal{A}^2$ that have been compared, 
and by $C = |\mathcal{C}| \geq 1$ its cardinal. 
For $a \in\mathcal{A}$, let $\mathcal{A}_a = \{ b \in \mathcal{A}, \ (ab) \in \mathcal{C}\} \subset \mathcal{A}$ be the set of alternatives that have been compared with $a$ and $A_a$ its cardinal. 
Then, $C \leq A(A-1)$, where the equality corresponds to the case where all the alternatives have been compared. 
If so, we will have that $\mathcal{A}_a = \mathcal{A} \backslash \{a\}$ and $A_a = A-1$ for any $a \in \mathcal{A}$. 

The comparisons are then characterized by an antisymmetric \emph{comparison matrix}\footnote{The comparison matrix is not a matrix in the classical sense, since its entries $r_{ab}$ are only defined for $(ab) \in \mathcal{C}$ and not $(ab) \in \mathcal{A}^2$ in general.}
\begin{equation}
    \Comp = (r_{ab})_{(ab) \in \mathcal{C}}. 
\end{equation} 
Our goal is to attribute a score $\theta_a$ to each alternative $a \in \mathcal{A}$ from the paired comparison data $\Comp$,
i.e. to construct a function $\Theta^* : \Comp \mapsto \Theta^*(\Comp) \in \R^{{A}}$
that yields desirable computational, monotonicity and resilience properties.

    \subsection{GBT Models as an Exponential Subfamily} \label{sec:defGBT}

To infer scores from comparisons, as is often done in probabilistic models, 
we first define a distribution of comparisons given scores.
Our proposal is to focus on a particular subset of the widely studied \emph{exponential family}~\cite{barndorff2014information}.
Recall that this family is of the form 
\begin{equation}
    p(r | \theta) = f(r) g(\theta) \exp (h(\theta) k(r))
\end{equation}
for some functions $f,g,h,k$, where $r$ is a random variable whose law is parameterized by $\theta$.
Essentially, we define GBT models as the subfamily of such models 
where $h(\theta) = \theta$ and $k(r) = r$.
In fact, since $g(\theta)$ is merely a normalization factor,
this amounts to defining $p(r|\theta) \propto f(r) \mathrm{e}^{r\theta}$.
Since the distribution $p(r|\theta)$ is fully determined by $f$, 
assuming it is normalized so that $\int_\R f(r)\mathrm{d}r = 1 = \int_\R \mathrm{d}F(r)$,
we call $f$ the \emph{root law} of the GBT model, and $F$ is the associated cdf.
More precisely, we define GBT models as follows.

\begin{definition}
    For any probability law $f$ over $\R$ with finite exponential moments
    (i.e. $\int_{\R} e^{r \theta} dF(r) < \infty$ for any $\theta \in \R$),    
    we define the \emph{$f$-GBT model} as follows.
    For any hidden scores $\theta \in \R^{ A}$, 
    the comparisons $\Comp_{ab}$ given $\theta$ are assumed to be independent
    and each only depends on the score difference $\theta_{ab} \triangleq \theta_a - \theta_b$,
    with $p(r_{ab} | \Theta) \propto f(r_{ab}) e^{r_{ab} \theta_{ab}}$. 
\end{definition}

Equivalently, this corresponds to defining, 
for any $(ab)\in \mathcal{C}$ and any $\Theta \in \R^A$,
    \begin{equation} \label{eq:cecicela}
        p(r_{ab}|\Theta) =\frac{\mathrm{e}^{\theta_{ab}r_{ab}} f(r_{ab})}{\int_{\R} \mathrm{e}^{\theta_{ab}r} \mathrm{d}F(r)},
    \end{equation}
which is well-defined when the \emph{root law} $f$ has finite exponential moments.
Note that, assuming $\theta_a = \theta_b$, 
the (normalized) \emph{root law} $f$ is then exactly the probability distribution of a comparison $r_{ab}$.
In other words, $f$ describes the distribution of comparisons between alternatives of similar quality.
This makes it natural to expect $f$ to be symmetric with respect to zero,
which implies that its support $\textsc{Supp}(f)$ is as well.
Using the independence of the comparisons conditionally to the scores, we easily deduce the following result. 

\begin{proposition} \label{prop:lawGBT}
    Under the $f$-GBT model,
    the comparisons $r_{ab}$ are independent conditionally to $\Theta \in\R^A$ and $r_{ab}| \Theta = r_{ab} |\theta_{ab}$. Moreover, we have, for any $(\Comp,\Theta)$
    \begin{equation} \label{eq:comptheta}
    p(\Comp|\Theta) =  
     \prod_{(ab)\in \mathcal{C}} p(r_{ab}|\theta_{ab}) 
     =
    \prod_{(ab)\in \mathcal{C}}  \frac{f\left( r_{ab} \right)   
    \mathrm{e}^{r_{ab} \theta_{ab}}}{\int_{\R} \mathrm{e}^{\theta_{ab} r} \mathrm{d}F(r)}
    \end{equation}
\end{proposition}


    \subsection{Cumulant-Generating Functions} \label{sec:logpartition}

Let us now introduce the cumulant-generating function~\cite{kenney1951cumulants} derived from the \emph{root law} $f$,
which will be central to our analysis of $f$-GBT models.

\begin{definition}
    Let $f$ be a probability law over $\R$ with finite exponential moment. Its \emph{cumulant generating function} is defined for any $\theta \in \R$ by
    \begin{equation} \label{eq:defphiab}
    \Phi_f(\theta) = \log \left( \int_{\R} \mathrm{e}^{\theta r} \mathrm{d}F(r)\right),
    \end{equation}
    where we recall that $F$ is the cdf of $f$. 
\end{definition}

The cumulant generating function of $f$, also called the log-partition function~\cite{wainwright2005new}, will play an important role for GBT models. By extension, we say that $\Phi_f$ is the cumulant generating function of the $f$-GBT model.

The Taylor series expansion of $\Phi_f$ provides the cumulants of $f$. The function $\Phi$ is known for many classical probability laws\footnote{\url{https://en.wikipedia.org/wiki/Moment-generating_function\#Examples}} and has been extensively studied, in particular in large deviations theory~\cite{dembo2009large}. We recall some of its main properties in Theorem~\ref{theo:Phi}, whose proof is given in the Appendix for the sake of completeness. We recall that $\Dc$ is the support of $f$ and we denote by $r_{max} = \sup \Dc \in (0, \infty] $.

 \begin{theorem}[\cite{dembo2009large}]
 \label{theo:Phi}
    For any root law $f$ with finite exponential moments, 
    the cumulant-generating function $\Phi_f$ is non-negative, strictly convex, even, and infinitely smooth over $\R$.
    Its derivative $\Phi_f'$ is a strictly increasing odd bijection from $\R$ to its image, which is $(- r_{\max} , r_{\max})$ as soon as $r_{\max} < \infty$.
\end{theorem}

As we will see in the next section,
the cumulant-generating function $\Phi_f$ yields numerous key statistics
of the maximum-a-posteriori score estimator.

        \subsection{Discrete and Continuous GBT Models}
    \label{sec:discreteandconti}

For concreteness, we distinguish two types of GBT models, depending on the discreteness of the comparisons' domain.
As we will see, the discreteness of $f$ defines the discreteness of $f$-GBT.



\paragraph{Discrete GBT models.} 
Let the root law $f$ be of the form
\begin{equation}
    f(r)= \sum_{k \in \mathcal{K}} \mathbb{P}[r_k|0] \delta_{r_k}(r),
\end{equation}
where $\delta_x$ is the Dirac distribution on $x$
and $\mathcal{K}$ is a countable (possibly finite) subset of $\R$ with no accumulation point.
The discrete root law has an associated probability mass function $\mathbb{P}[r|0]$.
The random variable $\Comp|\Theta$ is then also discrete with identical support $\Dc$. According to \eqref{eq:comptheta}, its probability mass function $P(\Comp|\Theta)$ is given by
\begin{equation} \label{eq:Pcompthetadiscrete}
    P(\Comp|\Theta) = \prod_{(ab)\in \mathcal{C}} \frac{\mathbb{P}[r_{ab}|0]  \mathrm{e}^{r_{ab}\theta_{ab}}}{ \sum_{r \in \Dc} \mathbb{P}[r|0]\mathrm{e}^{r\theta_{ab}}}. 
\end{equation}

\paragraph{Continuous GBT models.}
In the continuous setting, 
the comparison matrix $\Comp|\Theta$ also admits a probability density function given by 
\begin{equation} \label{eq:pcompthetacontinuous}
    p(\Comp|\Theta)= \prod_{(ab)\in \mathcal{C}} \frac{f(r_{ab}) \mathrm{e}^{r_{ab}\theta_{ab}}}{\int_{\Dc} \mathrm{e}^{r\theta_{ab}}f(r) \mathrm{d}r}. 
\end{equation}



    \subsection{Historical Bradley-Terry as Bernoulli-GBT} \label{sec:histoBT}

Recall that the classical~\cite{BradleyTerry52} model is characterized by the fact 
that comparisons follows the following distribution:\footnote{This is one possible parametrization, using an exponential score function. An alternative is to consider that $P(a > b) = \frac{p_a}{p_a+p_b}$ with $p_a > 0$ the score value of $a$.}
\begin{equation} \label{eq:oldBT}
    P ( r_{ab} = 1 | \theta_a - \theta_b) = \frac{\mathrm{e}^{\theta_a}}{\mathrm{e}^{\theta_a}+\mathrm{e}^{\theta_b}}
\end{equation}
and $P ( r_{ab} = 1 | \theta_a, \theta_b) + P ( r_{ab} = -1 | \theta_a, \theta_b) = 1$.

Here, we highlight the fact that this is an instance of the discrete GBT models. 
Namely, consider the Bernouilli \emph{root law}
$f = \frac{\delta_1 + \delta_{-1}}{2}$.
Its cumulant-generating function is
\begin{equation}
    \Phi_{\text{Bernoulli}}(\theta) = 
    \log \left( \frac{\mathrm{e}^{\theta} + \mathrm{e}^{-\theta}}{2}\right) = \log (\cosh \theta).
\end{equation}
Using \eqref{eq:cecicela} and $\theta_{ab}=\theta_a - \theta_b$, we deduce 
\begin{equation}
    P ( r_{ab} = 1 | \Theta ) = \frac{\mathrm{e}^{\theta_{ab}}}{2 \cosh(\theta_{ab})} = \frac{\mathrm{e}^{\theta_a}}{\mathrm{e}^{\theta_a}+\mathrm{e}^{\theta_b}}.
\end{equation}
We recover~\eqref{eq:oldBT},
hence the fact that \cite{BradleyTerry52} is the Bernoulli-GBT model.

        \section{MAP Scores Based on GBT Models}
        \label{sec:estimators}
        
In this section, we take a Bayesian approach to define a maximum-a-posteriori estimator based on GBT models,
and prove some of the basic resulting properties.

    \subsection{Bayesian Model for Score Estimation} 
    Consider a normal prior $\mathcal N(0, \sigma^2 I)$ on $\Theta$. 
    In the continuous setting, the posterior density function of $\Theta$ conditionally to the comparisons $\Comp$ is, using Bayes law,    
    \begin{equation} \label{eq:be}
        p(\Theta |\Comp) = \frac{p(\Comp| \Theta) p(\Theta)}{p(\Comp)}.
    \end{equation}
    In the discrete setting, probability density functions are replaced by probability mass functions.

    \textit{Remark.} Using \eqref{eq:pcompthetacontinuous} in \eqref{eq:be}, we can interpret the GBT model with a Gaussian prior (or a prior from any law of the exponential family) as an exponential family for the random variable $\Theta$, where $\Comp$ is the natural parameter. 
    
    \subsection{Maximum A Posteriori Estimator} 
   We define the \emph{negative log-posterior}, either for discrete or continuous comparisons,  as 
   \begin{equation} \label{eq:LBayes}
       \mathcal{L} (\Theta|\Comp) =  - \log p (\Theta |\Comp).
   \end{equation}
   Interestingly, this yields a loss function that directly depends on the cumulant-generating function, as
   \begin{equation}
   \label{eq:total_loss}
       \mathcal{L} (\Theta|\Comp) =  \frac{1}{2\sigma^2} \sum_{a \in \mathcal{A}}\theta_a^2 + \sum_{(ab) \in \mathcal{C}} \left( \Phi_f(\theta_{ab}) - r_{ab} \theta_{ab} \right).
   \end{equation}
   In particular, the nice properties of $\Phi_f$ 
   and the fact that the Gaussian prior is turned into a quadratic regularization
   imply that the negative log-posterior $\mathcal{L}$ is well-behaved.
   
       \begin{proposition} \label{prop:existuniquemin}
For any comparison matrix $\Comp$,
the negative log-posterior $\mathcal{L}(\cdot |\Comp)$ is $(1/\sigma^2)$-strongly convex, and thus admits a unique minimizer $\Theta^*_{f, \sigma^2}  (\Comp) \in \R^{{A}}$.
\end{proposition}

Since it is the mode of the posterior, $\Theta^*_{f, \sigma^2}(\Comp)$ is commonly known as the \emph{maximum a posteriori (MAP) estimator}. 
Crucially, for any $f$-GBT model with a normal prior, Proposition~\ref{prop:existuniquemin} guarantees that the MAP is well-defined and fastly computable by any strongly convex optimizer.
In fact, it is used in Tournesol to estimate individual scores from user comparisons~\cite{Pipeline2023}.

    \subsection{First Properties of Score Estimators}

In the following, 
we show that the maximum likelihood score estimator $\Theta^*_{f, \sigma^2}(\Comp)$ has zero mean in  Proposition~\ref{prop:bayesianconstraint}, 
we provide its first two moments in Proposition \ref{prop:momentstheta}, 
and we give a bound on its supremum norm in Proposition~\ref{prop:boundinfinity}. 
The proofs are in the appendix.

\begin{proposition} \label{prop:bayesianconstraint}
    For any comparison matrix $\Comp$, 
    the MAP estimator 
    $\theta^* = \Theta^*_{f, \sigma^2} (\Comp)$ verifies $ \sum_{a \in \mathcal{A}}  \theta_a^*  = 0$.
\end{proposition}

We denote by $\mathbb{E} [g(\Comp)|\Theta]$ and $\mathbb{V}[g(\Comp)|\Theta]$ the mean and the variance of $g(\Comp)$ conditionally to the scores. 
Also, denote by $\mathbb{C}\mathrm{ov}[g(\Comp),h(\Comp)|\Theta] $
the covariance of $g(\Comp)|\Theta$ and $ h(\Comp)|\Theta$.

\begin{proposition} 
    \label{prop:momentstheta}
    Let $(ab), (cd) \in \mathcal{C}$ with $(cd) \notin \{(ab), (ba)\}$.
    Then,
    \begin{align} \label{eq:meanvariance}
        &\mathbb{E}[r_{ab}|\Theta] = \Phi'_f(\theta_{ab})  \quad  \text{and}  \\ 
        &\mathbb{V}[r_{ab}|\Theta] = \Phi''_f(\theta_{ab}), \quad 
        \mathbb{C}\mathrm{ov}[r_{ab}, r_{cd} |\Theta ] = 0. \nonumber
    \end{align}
\end{proposition}

\begin{proposition}
    \label{prop:boundinfinity}
    If $r_{\max} < \infty$, then the MAP estimator $\Theta^*_{f, \sigma^2}(\Comp)$ verifies, for any $a \in \mathcal{A}$,
    \begin{equation}
        |\theta_a^*(\Comp)| \leq 2 A_a r_{\max} \sigma^2,
    \end{equation}
    and therefore
    \begin{equation}
    \|\Theta^*_{f, \sigma^2}(\Comp) \|_\infty \leq  2 r_{\max}    \sigma^2 \sup_{a \in \mathcal{A}} A_a \leq  2 r_{\max}  \sigma^2 (A-1). 
    \end{equation}
\end{proposition}

\section{Monotonicity of MAP Score  Estimators}\label{sec:monotonicity}

In this section, we prove a desirable property of MAP estimators 
for all GBT models with a Gaussian prior.
Namely, we show that the more an alternative wins comparisons, the better it is scored.

    \subsection{Partial Orders over Paired Comparisons}
    \label{sec:partialorders}

Let us first formalize a partial order between comparisons associated to a given alternative, 
which captures the idea that $a$ is better compared when all the other comparisons between alternatives different from $a$ are fixed. 
Note that this partial order is only defined for comparison matrices $\Comp$ and $\Comp'$ sharing the same set of compared pairs $\mathcal{C}$. 

\begin{definition} \label{def:partialorder}
    For any alternative $a \in \mathcal{A}$, we say that two comparison matrices $\Comp$ and $\Comp'$ satisfy 
    \begin{equation}
        \Comp \leq_a \Comp'
    \end{equation}
    if (i)  $\comparison_{a b} \leq \comparison_{a b}'$ for all $b$ such that $(ab) \in \mathcal{C}$ and (ii)  $ \comparison_{cd} = \comparison_{cd}'$ for $(cd) \in \mathcal{C}$ with $a \notin \{ c,d \}$. The relation 
    \begin{equation}
        \Comp <_a \Comp'
    \end{equation}
    means that (i) is strict for at least one $b$ such that  $(ab) \in \mathcal{C}$. 
\end{definition}

We formalize the notion that a score estimator is consistent with the partial orders of Definition~\ref{def:partialorder}.

\begin{definition} \label{def:increasingscoring}
    We say that an estimator $\widehat{\Theta}(\Comp)$ is \emph{increasing} with respect to the $\Comp$ if, for any $a \in \mathcal{A}$, 
    \begin{equation}
        \Comp \leq_a \Comp' \Longrightarrow \widehat{\theta}_a (\Comp) \leq \widehat{\theta}_a(\Comp'),
    \end{equation}
    and \emph{strictly increasing} with respect to $\Comp$ if, for any $a \in \mathcal{A}$, 
    \begin{equation}
        \Comp <_a \Comp' \Longrightarrow \widehat{\theta}_a (\Comp) < \widehat{\theta}_a (\Comp'). 
    \end{equation} 
\end{definition}

    \subsection{Elementary Monotonicity Criteria}

We provide criteria for the monotonicity of score estimators for continuous and discrete comparisons, in Proposition~\ref{prop:onpartialrab} and Proposition~\ref{prop:ondeltarabdiscrete} that are proved in the Appendix.

\begin{proposition} \label{prop:onpartialrab}
    We suppose that $\Comp$ is continuous-domain and that the estimator $\widehat{\Theta}(\Comp)$ is differentiable with respect to $r_{ab}$ for any $(ab) \in \mathcal{C}$. Then, $\widehat{\Theta}(\Comp)$ is increasing with respect to $\Comp$ if and only if, for any $(ab) \in \mathcal{C}$ and any $\Comp$,
    \begin{equation} \label{eq:partialtheta}
        \partial_{r_{ab}} \widehat{\theta}_a(\Comp) \geq 0.
    \end{equation}
    It is moreover strictly increasing if and only if the inequalities are strict.
\end{proposition}

For discrete comparisons, we define finite-difference operators over scores as follows.
For any $\Comp$ and $(ab) \in \mathcal{C}$, we define $\Comp'_{ab}$ as the comparison with identical scores, except at position $(ab)$ where the score is increased from $r_k$ to $r_{k+1}$ (or remains unchanged if $r_K$ is reached) and reduces symmetrically the comparison at position $(ba)$. Then, we define the operator $\Delta_{(ab)}$ over functions $F : \Dc^C \rightarrow \R$,
\begin{equation} \label{eq:defdeltaab}
    \Delta_{(ab)} F (\Comp) = \frac{F(\Comp^{\mathrm{up}}_{(ab)}) - F(\Comp)}{r_{k+1} - r_k}
\end{equation}
if $r_k < r_{K}$ and $\Delta_{(ab)} F (\Comp) = 0$ if $r_k = r_K$.

\begin{proposition} \label{prop:ondeltarabdiscrete}
    We suppose that $\Comp$ is discrete-domain. Then, $\widehat{\Theta}(\Comp)$ is increasing with respect to $\Comp$ if and only if, for any $(ab) \in \mathcal{C}$ and any $\Comp$,
    \begin{equation} \label{eq:partialthetadiscrete}
        \Delta_{(ab)} \widehat{\theta}_a(\Comp) \geq 0.
    \end{equation}
    It is moreover strictly increasing if and only if the inequalities are strict.
\end{proposition}

    \subsection{Monotonicity of GBT Estimators}

    We show that the monotonicity of GBT estimators is automatically satisfied.  The proof is in the Appendix. 


\begin{theorem} \label{theo:monotone}
For any $f$-GBT model with a Gaussian prior, 
the MAP estimator $\Theta^* (\Comp)$ is 
strictly increasing with $\Comp$ in the sense of Definition~\ref{def:increasingscoring}. 
\end{theorem}

The proof relies on the monotonicity criteria for continuous (Proposition~\ref{prop:onpartialrab}) and discrete (Proposition~\ref{prop:ondeltarabdiscrete}) GBT models. In order to show that 
$\partial_{r_{ab}} \widehat{\theta}_a(\Comp) \geq 0$, 
we analyze the gradient relation 
$\nabla \mathcal{L}(\Theta^*_{f, \sigma^2}(\Comp)|\Theta) = 0$. By applying $\partial_{\theta_{ab}}$ to it, we obtain a linear system on $\partial_{\theta_{ab}}\Theta^*_{f, \sigma^2}(\Comp)$. The study of this linear system relies on the properties of diagonally-dominant matrices and leads to the desired result.

Note that our proof yields a slightly more general result, 
as the monotonicity actually holds for all coordinate-independent priors
(i.e. $\theta_a$ is a priori independent from $\theta_b$ for $a \neq b$)
which yield a strongly convex negative log-prior.
In fact, it can be extended to convex negative log-prior, 
or even to the maximum likelihood estimator (i.e. no prior),
if we consider inferred scores with values in $[-\infty, +\infty]$.

    \subsection{Impact of New Comparisons}
    
     Assume that $\Comp$ and $\Comp'$ are two comparison matrices over $\mathcal{C}$ and $\mathcal{C}'$ respectively, such that
    \begin{equation}
        \mathcal{C}' = \mathcal{C} \cup \{ (ab), (ba) \} 
    \end{equation}
where $(ab) \notin \mathcal{C}$. We assume that $r_{cd}=r_{cd}'$ for any $(cd) \in \mathcal{C}$, hence $\mathcal{C}$ and $\mathcal{C}'$ only differs by the new comparison. We evaluate the impact of this comparison $r_{ab}'$ over the score vector $\Theta^*_{f, \sigma^2}(\Comp')$. 

\begin{proposition} \label{prop:equalscores}
    For $\Comp$ and $\Comp'$ as defined above, we have the equivalences
    \begin{align} \label{eq:equivalencerab}
        \Theta^*_{f, \sigma^2}(\Comp) = \Theta^*_{f, \sigma^2}(\Comp') &\Longleftrightarrow r_{ab}' = \Phi_f'(\theta_{ab}^*(\Comp)), \\
        \Theta^*_{f, \sigma^2}(\Comp) < \Theta^*_{f, \sigma^2}(\Comp') &\Longleftrightarrow r_{ab}' > \Phi_f'(\theta_{ab}^*(\Comp)), \\
        \Theta^*_{f, \sigma^2}(\Comp) > \Theta^*_{f, \sigma^2}(\Comp') &\Longleftrightarrow r_{ab}' < \Phi_f'(\theta_{ab}^*(\Comp)). 
    \end{align}
\end{proposition}

The proof is provided in the Appendix.
Now, $\Phi_f'$ can be interpreted as a conversion function between score differences $\theta_{ab}$ and comparisons $r_{ab}$. This is highlighted by the relation~\eqref{eq:equivalencerab} and the fact that $\mathbb{E}[r_{ab} | \Theta ] = \Phi_f'(\theta_{ab})$ (Proposition~\ref{prop:momentstheta}).
When $r_{\max} = 1$, we observe that $\Phi_f'$ is a sigmoid function, i.e. an increasing bijection from $\R$ to $(-1,1)$.

\begin{table*}[ht]
\centering
\begin{tabular}{||c c c c c ||}  
 \hline
 GBT model & Parameter & $F/f(r|0)$ & $\Dc$ &   $\Phi(\theta)$    \\ [0.5ex] 
 \hline\hline
 Binary & - & $\frac{1}{2}$  & $\{\pm 1\}$  & $\log\left( \cosh \theta \right)$   \\ [1ex]  
 \hline  
 $K$-nary & $K \geq 2$  & $\frac{1}{K} $  & $\{ -1 , \ldots, +1\} $ & $\log\left( \frac{\sinh \left( \frac{K\theta}{K-1}\right) }{K \sinh \left(  \frac{\theta}{K-1} \right)}\right) $   \\  [1ex] 
 \hline  
 Poisson & $\lambda > 0$ & $
      \frac{\mathrm{e}^{-\lambda} \lambda^{|k|}}{ |k| !} (  \delta_{k=0} + \frac{\delta_{k\neq 0}}{2} )$  & $\mathbb{Z}$  & $\lambda \cosh(\theta)$  \\  [1ex] 
 \hline   Gaussian & $\sigma_0^2 > 0$ & $g_{\sigma_0^2}(r) = \frac{1}{\sqrt{2\pi} \sigma_0} \mathrm{e}^{-r^2 / 2 \sigma_0^2}$  & $\R$  & $\frac{\sigma_0^2 \theta^2}{2}$  \\  [1ex] 
 \hline  
Beta & $\beta > 0$ & $\frac{\Gamma(2\beta)}{4^{\beta} \Gamma(\beta)^2} (1-r^2)^{\beta-1}$  & $(-1,1)$ & $ \log \left(  1 + \sum_{k \geq 1} \left( \prod_{n=0}^k \frac{\beta + n}{2\beta + n} \right) \frac{\theta^{2k}}{(2k)!}  \right)$      \\ [1ex] 
 Uniform & $\beta = 1$ & $\frac{1}{2}$  & $(-1,1)$  & $\log \left( \frac{\sinh \theta}{\theta}\right)$  \\ [1ex]  
   & $\beta = 2$ & $\frac{3 }{4}(1-r^2)$  & $(-1,1)$  & $ \log \left( \frac{3 (\theta \cosh \theta - \sinh \theta ))}{\theta^3}\right) $   \\ [1ex] 
  \hline  
 \end{tabular}  
 \label{table:GBT_models}
 \caption{Examples of GBT models, with their cumulant-generating functions.}
\end{table*}

\section{Lipschitz-Resilience of Score MAP Estimators}    \label{sec:resilience}

    Among the motivation for the generalized Bradley-Terry model, we aim at a scoring method from expressed pair comparisons which controlled impact of user's decision. In the more global Tournesol pipeline~\cite{tournesol}, the individual scoring method is used at a first step in a global scoring method for global scoring from individual comparisons~\cite{AllouahGHV22}. 
    We formalize the notion of Lipschitz-resilience to the user's updates and provide criteria to determine if a given GBT model is resilient. We show that the Lipschitz-resilience is guaranteed as soon as the comparison domain $\Dc$ is bounded.

    \subsection{Lipschitz-Resilience to User Modifications}

    The Lipschitz-resilience of an estimator captures its ability to be limitedly modified by changing or adding new paired comparisons for the user.\\

    Let $\Comp \in \Dc^{C}$ and $\Comp' \in \Dc^{C'}$ be two comparison matrices over some possibly distinct $\mathcal{C}$ and $\mathcal{C}'$ of respective size $C$ and $C'$. We define the \emph{symmetric difference} of $\mathcal{C}$ and $\mathcal{C}'$ as $\mathcal{C} \Delta \mathcal{C}' = (\mathcal{C} \backslash \mathcal{C}') \cup (\mathcal{C} \backslash \mathcal{C}')$ and denote its cardinal by
    \begin{equation} 
        \Delta_{\mathrm{domain}}(\Comp,\Comp') = |\mathcal{C}\Delta\mathcal{C}'|.
    \end{equation}
    The set $\mathcal{C} \Delta \mathcal{C}'$ is made of pairs $(ab)$ that are in one of the two set $\mathcal{C}$ and $\mathcal{C}'$ and not on the other. The number $\Delta_{\mathrm{domain}}(\Comp,\Comp')$ therefore quantifies the number of comparisons needed to be added or removed to transform $\mathcal{C}$ into $\mathcal{C}'$. 

    We define the matrices $\tilde{\Comp} = (r_{ab})_{(ab) \in \mathcal{C} \cap \mathcal{C}'}$ and $\tilde{\Comp}' = (r_{ab}')_{(ab) \in \mathcal{C} \cap \mathcal{C}'}$, both in $\Dc^{|\mathcal{C}\cap  \mathcal{C}'|}$, which coincide with $\Comp$ and $\Comp'$ on $\mathcal{C}\cap\mathcal{C}'$. 
    We recall that the \emph{$L0$ ``norm"} of a vector is the number of its non-zero entries. We then define 
    \begin{equation}
        \Delta_{\mathrm{entries}}(\Comp,\Comp') = \| \tilde{\Comp} - \tilde{\Comp}' \|_0,
    \end{equation}
    which measures the number of entries $(ab) \in \mathcal{C} \cap \mathcal{C}'$ on which $\Comp$ and $\Comp'$ differ. We also set 
    \begin{equation}
        \Delta( \Comp, \Comp')
        =
         \Delta_{\mathrm{domain}} (\Comp,\Comp') + \Delta_{\mathrm{entries}} (\Comp,\Comp'),
    \end{equation}
    which counts the number of elementary modifications (removing, adding, or changing a comparison) from $\Comp$ to $\Comp'$.

    \begin{definition}
        An estimator $\widehat{\Theta}(\Comp)$ is said to be \emph{$L$-Lipschitz-resilient} for some $L > 0$ for the Euclidean norm if, for any comparison matrices $\Comp ,\Comp'$,
        \begin{equation} \label{eq:widehattheta}
            \| \widehat{\Theta}(\Comp) - \widehat{\Theta}(\Comp') \|_2 \leq L \Delta( \Comp, \Comp')
        \end{equation}
    \end{definition}

If $\widehat{\Theta}$ is $L$-Lipschitz-resilient, then $L$ bounds the possible impact on the score $\widehat{\Theta} (\Comp)$ for single modifications of the comparison matrix (single update over one comparison or addition of a new comparison). 
This impact is measured in terms of $\ell_2$-norm. 
When the comparison sets $\mathcal{C}=\mathcal{C}'$ coincide\footnote{This ensures that $\Comp-\Comp'$ is well-defined.}, the bound \eqref{eq:widehattheta} is simply 
\begin{equation}
    \| \widehat{\Theta}(\Comp) - \widehat{\Theta}(\Comp') \|_2 \leq L \| \Comp - \Comp' \|_0.
\end{equation}

    \subsection{Lipschitz-Resilience Guarantee for GBT Models} \label{sec:relientornot}
 
    \begin{theorem} \label{theo:resilient} 
        For any root law $f$ and given a Gaussian prior $\mathcal N(0, \sigma^2 I)$,
        the MAP estimator for the $f$-GBT model is $(4 \sqrt{2}r_{\max}\sigma^2)$-Lipschitz-resilient, i.e.
        \begin{equation} \label{eq:boundresilience}
            \sup_{\Comp \neq \Comp'} \frac{\| \Theta_{f, \sigma^2}^*(\Comp) - \Theta_{f, \sigma^2}^*(\Comp') \|_2}{\Delta( \Comp, \Comp')} \leq   4 \sqrt{2}r_{\max} \sigma^2.
        \end{equation}
    \end{theorem}

   In the GBT model, contrary to the historical BT model, we regularize the scores using a Gaussian prior on $\Theta$. We see in Theorem~\ref{theo:resilient} that the Lipschitz-resilience constant explodes when $\sigma^2 \rightarrow \infty$ (i.e. with no regularization). The regularization leads to controllable user's modifications, where the prior variance $\sigma^2$ plays a crucial role.
   Moreover, the Gaussian and Poisson BT models (see examples below), for which the comparison domain $\Dc = \R$ or $\mathbb{Z}$ is unbounded, are not Lipschitz-resilient (we provide a proof in the Appendix). The boundedness of the comparisons is a key ingredient to the Lipschitz-resilience. 
    
\section{Examples of GBT Models} \label{sec:examples}

The examples detailed in this section are listed in Table~1, together with their corresponding comparison domain and their cumulant-generating function. 
Each of these models can be used to provide score estimators based on paired comparisons that are all strictly increasing with respect to $\Comp$ according to Theorem~\ref{theo:monotone}. They are moreover all resilient for bounded comparisons, while the Gaussian and Poisson-GBT model are not (see Proposition~\ref{prop:gaussianproperties} below and the appendix).

\subsection{The Gaussian-GBT Model}

The Gaussian GBT model is characterized by a Gaussian root law $f = g_{\sigma^2_0}$. 
This model has already been studied \cite{guo2012score} and applied~\cite{kristof2019user}. We summarize its main properties in Proposition~\ref{prop:gaussianproperties}. The proof and the closed form expression of $\Theta^*(\Comp)$ are given in the appendix.

\begin{proposition}
    \label{prop:gaussianproperties}
    The Gaussian-GBT model with variance $\sigma^2_0$ is such that 
        \begin{equation} \label{eq:lawcomptthetagaussfirst}
        \Comp | \Theta \sim \mathcal{N} \left( (\sigma_0 \theta_{ab})_{(ab)\in \mathcal{C}} , \sigma_0^2 \mathrm{Id} \right).
    \end{equation}
    The MAP estimator $\Theta^*(\Comp)$ is linear, strictly increasing with respect to $\Comp$, and is \emph{not} Lipschitz-resilient.
\end{proposition}

\subsection{The Uniform-GBT Model}

The Uniform-GBT model corresponds to choosing the uniform probability law $f=1/2$ on $[-1,1]$ as the root law of the model.  We recall that the prior variance is denoted by $\sigma^2$. The Uniform-GBT model is used for the current version of the Tournesol pipeline~\cite{Pipeline2023}. 

\begin{proposition}
    \label{prop:uniformproperties}
    The cumulant generating function of the Uniform-GBT model is $\Phi(\theta) = \log (\sinh(\theta) / \theta)$.
    The MAP estimator $\Theta^*(\Comp)$  is strictly increasing with respect to $\Comp$ and $4\sqrt{2}\sigma^2$-Lipschitz-resilient.
\end{proposition}

\textit{Remark.} The derivative $\Phi'(\theta) = \coth(\theta) - \frac{1}{\theta}$
of $\Phi$ is known as the Langevin function \cite{cohen1991pade}.

\section{Empirical Simulations} 

\begin{figure*}[!ht]
    \centering
    \includegraphics[width=53mm]{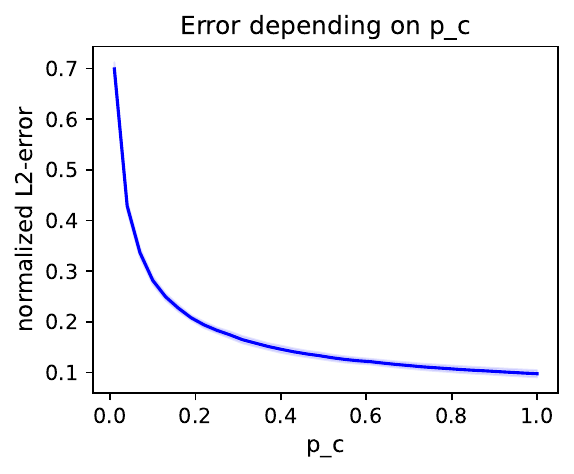}%
    \includegraphics[width=54mm]{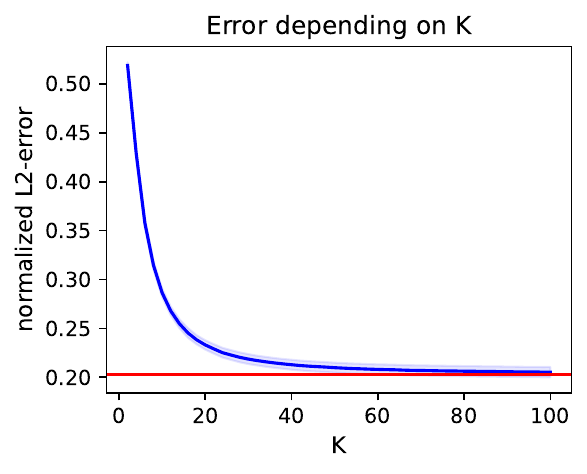}   \includegraphics[width=56mm]{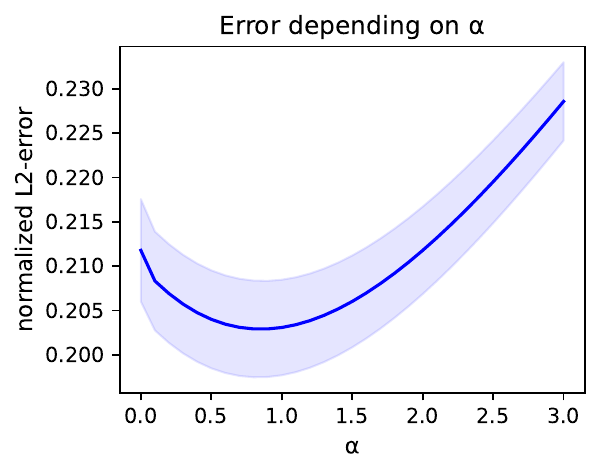}%
    \caption{Left: Normalized mean-square error with respect to the sparsity parameter $p_c$ for Erdös-Rényi comparison graphs. Middle: Normalized mean-square errors $\textsc{NormError}_K$ using $K$-nary-GBT MAP estimators on the data generated via the Uniform-GBT model. Blue: $K \mapsto \textsc{NormError}_K$; Red: $\textsc{NormError}$ for the Uniform-GBT model. Right: Normalized mean-square error with respect to the regularization scale $\frac{1}{\sigma^2}$.}
    \label{fig:experiments-1}
    \vspace{-15pt}
\end{figure*}

We propose three experiments that illustrate interesting properties of GBT models. The generative data model is itself a GBT model, and does not simulate a realistic situation corresponding to real data. These simulations allow us to measure three aspects of the model: (i) the impact of the sparsity of the comparison graph, (i) the impact of the discretization level, and (iii) the impact of the regularization parameter (prior variance) on the quality of the reconstruction. The results are expressed in terms of the normalized mean-square error against the true scores $\Theta^\dagger$, given by
\begin{equation}
    \textsc{NormError} = \mathbb{E} \left[ \frac{ \|  {\Theta(\Comp) - \Theta^\dagger}\|_2^2}{{\|\Theta^{\dagger}\|_2^2}} \right]
\end{equation}
\urlstyle{tt}
We use Monte-Carlo simulations to obtain normalized mean-square errors for various GBT MAP estimators with Gaussian prior. 
Plots in Figure~\ref{fig:experiments-1} depict the mean and standard deviation of the Monte-Carlo simulations. We consider $A = 500$ alternatives. 
The ground-truth scores $\Theta^\dagger \in \R^{500}$ are generated as i.i.d. Gaussian random variables with variance $\sigma^{\dagger 2} = 1$. 
 All experiments are run
with ten seeds from $1$ to $10$. The code and details of the experiments are available at {\small \url{https://github.com/sadeghfarhadkhani/GBT}}.

\paragraph{(i) Impact of the graph sparsity.}
We generate the comparisons $\Comp | \Theta^{\dagger}$  using the Uniform-GBT model over $[-1,1]$. The comparison set $\mathcal{C}$ is generated as an Erdös-Rényi random graph $\mathcal{G}(A = 500,p_c)$ where the nodes are different alternatives and each edge corresponds to a comparison which is randomly activated independently from other edges with probability $p_c \in [0,1]$.
We estimate the normalized mean-square error of MAP estimators based on the Uniform-GBT model with variance $\sigma^2 = 1$ for different values of $p_c$. The results are depicted in Figure~\ref{fig:experiments-1} (left). With no surprise, the sparsity of the graph strongly impacts the reconstruction performance.

\paragraph{(ii) $K$-nary-GBT models against Uniform-GBT model.}

The data are generated using the Uniform-GBT model on a comparison graph following the Erdös-Rényi random graph $\mathcal{G}(A = 500, p_c = 0.2)$. 
For any integer $K\geq 2$, we estimate $\Theta_K(\Comp)$ as the MAP estimator of the $K$-nary-GBT model with variance $\sigma^2=1$. We also compute $\Theta(\Comp)$ as the MAP estimator of the Uniform-GBT model with the same variance. We show the evolution of the normalized mean-square error with respect to $K$ in Figure~\ref{fig:experiments-1} (middle). 
The error decays with respect to $K$, and converges to the limit value corresponding to the Uniform-GBT reconstruction model. These results advocate for discretized comparison models when discretization is at play. The value $K=21$ was chosen on the Tournesol platform~\cite{Pipeline2023}, as a compromise between greater finesse on the comparisons and a restricted choice of possible comparisons for users.

\paragraph{(iii) Impact of the prior variance over the mean-square error.}
Recall from~\eqref{eq:total_loss} that
the regularization factor is scaled by $ \frac{1}{\sigma^2}$   where $\sigma^2$ is the variance of the prior for the score estimation model. Again, we use the Uniform-GBT model on a comparison graph following the Erdös-Rényi random graph $\mathcal{G}(A = 500, p_c = 0.2)$, to generate the data.
We then estimate the normalized mean-square error of MAP
estimators based on the Uniform-GBT model for different values of $\sigma$.
The results are given in Figure~\ref{fig:experiments-1} (right). 
We observe that the error for small $\frac{1}{\sigma^2}$ (i.e., little regularization) is slightly better than for $\sigma = \infty$ ($\frac{1}{\sigma^2} = 0$, i.e., no regularization), which is in favor of using a Gaussian Bayesian prior for the scores. This is a second advantage of the Bayesian approach in addition to Lipschitz-resilience.
A natural conjecture that our plot raises, which we leave open,
is whether the optimal prior variance for score inference matches the variance of the score distribution,
i.e. $\sigma = \sigma^\dagger$.

\section{Conclusion and Future Perspective} \label{sec:conclusion}

In this paper, we generalized the historical Bradley-Terry model, by defining a family of so-called GBT (generalized Bradley-Terry) probabilistic models, each of which transforms comparisons into scores.
This family is parameterized by a \emph{root law}, which models how two equally good alternatives are expected to be compared.
Remarkably, for any such prior, we proved that the derived GBT model is guaranteed to feature numerous desirable properties,
such as strict convexity, monotonicity and Lipschitz-resilience.
Because of these compelling features, GBT models seem desirable to deploy in many practical applications,
as they have been on Tournesol~\cite{tournesol} to turn contributors' comparisons of video recommendability into scores.

Having said this, GBT models raise further intriguing research questions.
For one thing, there may be other appealing models that also have the convexity, monotonicity and Lipschitz-resilience properties.
We leave open the problem of identifying the set of all models with such properties.
We also leave open the analysis of the statistical error of the GBT scores.
A related question is that of estimating the uncertainty on the GBT scores, given comparisons.
We stress that these analyses are challenging as they depend on the graph $\mathcal{C}$ of comparisons.
In fact, another related question is that of optimizing the graph to minimize the statistical error,
which is known as the \emph{active learning} problem.
Finally, a general analysis of how the GBT scores depend on the root law is also left open.

\newpage 

\bibliography{references}

\begin{thebibliography}{27}
\providecommand{\natexlab}[1]{#1}

\bibitem[{Allouah et~al.(2022)Allouah, Guerraoui, Hoang, and
  Villemaud}]{AllouahGHV22}
Allouah, Y.; Guerraoui, R.; Hoang, L.; and Villemaud, O. 2022.
\newblock Robust Sparse Voting.
\newblock \emph{CoRR}, abs/2202.08656.

\bibitem[{Anonymous(2023)}]{Pipeline2023}
Anonymous. 2023.
\newblock Permissionless Collaborative Algorithmic Governance with Security
  Guarantees.
\newblock \emph{Under submission}.

\bibitem[{Barndorff-Nielsen(2014)}]{barndorff2014information}
Barndorff-Nielsen, O. 2014.
\newblock \emph{Information and exponential families: in statistical theory}.
\newblock John Wiley \& Sons.

\bibitem[{Bradley and Terry(1952)}]{BradleyTerry52}
Bradley, R.~A.; and Terry, M.~E. 1952.
\newblock Rank analysis of incomplete block designs: I. The method of paired
  comparisons.
\newblock \emph{Biometrika}, 39(3/4): 324--345.

\bibitem[{Cattelan(2012)}]{cattelan2012models}
Cattelan, M. 2012.
\newblock Models for paired comparison data: A review with emphasis on
  dependent data.
\newblock \emph{Statistical Science}, 27(3): 412--433.

\bibitem[{Cohen(1991)}]{cohen1991pade}
Cohen, A. 1991.
\newblock A Pad{\'e} approximant to the inverse Langevin function.
\newblock \emph{Rheologica acta}, 30: 270--273.

\bibitem[{David(1963)}]{david1963method}
David, H.~A. 1963.
\newblock \emph{The method of paired comparisons}, volume~12.
\newblock London.

\bibitem[{Davidson(1970)}]{davidson1970extending}
Davidson, R.~R. 1970.
\newblock On extending the Bradley-Terry model to accommodate ties in paired
  comparison experiments.
\newblock \emph{Journal of the American Statistical Association}, 65(329):
  317--328.

\bibitem[{Dembo and Zeitouni(2009)}]{dembo2009large}
Dembo, A.; and Zeitouni, O. 2009.
\newblock \emph{Large deviations techniques and applications}, volume~38.
\newblock Springer Science \& Business Media.

\bibitem[{Durrett(2019)}]{durrett2019probability}
Durrett, R. 2019.
\newblock \emph{Probability: theory and examples}, volume~49.
\newblock Cambridge university press.

\bibitem[{Glickman(1999)}]{glickman1999parameter}
Glickman, M.~E. 1999.
\newblock Parameter estimation in large dynamic paired comparison experiments.
\newblock \emph{Journal of the Royal Statistical Society Series C: Applied
  Statistics}, 48(3): 377--394.

\bibitem[{Guo et~al.(2012)Guo, Sanner, Graepel, and Buntine}]{guo2012score}
Guo, S.; Sanner, S.; Graepel, T.; and Buntine, W. 2012.
\newblock Score-based Bayesian skill learning.
\newblock In \emph{Machine Learning and Knowledge Discovery in Databases:
  European Conference, ECML PKDD 2012, Bristol, UK, September 24-28, 2012.
  Proceedings, Part I 23}, 106--121. Springer.

\bibitem[{Herbrich, Minka, and Graepel(2006)}]{herbrich2006trueskill}
Herbrich, R.; Minka, T.; and Graepel, T. 2006.
\newblock TrueSkill™: a Bayesian skill rating system.
\newblock \emph{Advances in neural information processing systems}, 19.

\bibitem[{Hoang et~al.(2021)Hoang, Faucon, Jungo, Volodin, Papuc, Liossatos,
  Crulis, Tighanimine, Constantin, Kucherenko, Maurer, Grimberg, Nitu, Vossen,
  Rouault, and El{-}Mhamdi}]{tournesol}
Hoang, L.; Faucon, L.; Jungo, A.; Volodin, S.; Papuc, D.; Liossatos, O.;
  Crulis, B.; Tighanimine, M.; Constantin, I.; Kucherenko, A.; Maurer, A.;
  Grimberg, F.; Nitu, V.; Vossen, C.; Rouault, S.; and El{-}Mhamdi, E. 2021.
\newblock Tournesol: {A} quest for a large, secure and trustworthy database of
  reliable human judgments.
\newblock \emph{CoRR}, abs/2107.07334.

\bibitem[{Horn and Johnson(2012)}]{horn2012matrix}
Horn, R.~A.; and Johnson, C.~R. 2012.
\newblock \emph{Matrix analysis}.
\newblock Cambridge university press.

\bibitem[{Hunter(2004)}]{hunter2004mm}
Hunter, D.~R. 2004.
\newblock MM algorithms for generalized Bradley-Terry models.
\newblock \emph{The annals of statistics}, 32(1): 384--406.

\bibitem[{Karlis and Ntzoufras(2009)}]{karlis2009bayesian}
Karlis, D.; and Ntzoufras, I. 2009.
\newblock Bayesian modelling of football outcomes: using the Skellam's
  distribution for the goal difference.
\newblock \emph{IMA Journal of Management Mathematics}, 20(2): 133--145.

\bibitem[{Kenney and Keeping(1951)}]{kenney1951cumulants}
Kenney, J.; and Keeping, E. 1951.
\newblock Cumulants and the cumulant-generating function.
\newblock \emph{Mathematics of Statistics, Princeton, NJ}.

\bibitem[{Kristof et~al.(2019)Kristof, Quelquejay-Lecl{\`e}re, Zbinden,
  Maystre, Grossglauser, and Thiran}]{kristof2019user}
Kristof, V.; Quelquejay-Lecl{\`e}re, V.; Zbinden, R.; Maystre, L.;
  Grossglauser, M.; and Thiran, P. 2019.
\newblock A User Study of Perceived Carbon Footprint.
\newblock \emph{arXiv preprint arXiv:1911.11658}.

\bibitem[{Maher(1982)}]{maher1982modelling}
Maher, M.~J. 1982.
\newblock Modelling association football scores.
\newblock \emph{Statistica Neerlandica}, 36(3): 109--118.

\bibitem[{Maystre, Kristof, and Grossglauser(2019)}]{maystre2019pairwise}
Maystre, L.; Kristof, V.; and Grossglauser, M. 2019.
\newblock Pairwise comparisons with flexible time-dynamics.
\newblock In \emph{Proceedings of the 25th ACM SIGKDD International Conference
  on Knowledge Discovery \& Data Mining}, 1236--1246.

\bibitem[{Pe{\~n}a(1995)}]{pena1995m}
Pe{\~n}a, J.~M. 1995.
\newblock M-matrices whose inverses are totally positive.
\newblock \emph{Linear algebra and its applications}, 221: 189--193.

\bibitem[{Rao and Kupper(1967)}]{rao1967ties}
Rao, P.; and Kupper, L.~L. 1967.
\newblock Ties in paired-comparison experiments: A generalization of the
  Bradley-Terry model.
\newblock \emph{Journal of the American Statistical Association}, 62(317):
  194--204.

\bibitem[{Thurstone(1927)}]{thurstone1927law}
Thurstone, L.~L. 1927.
\newblock A law of comparative judgment.
\newblock \emph{Psychological review}, 34(4): 273.

\bibitem[{Wainwright, Jaakkola, and Willsky(2005)}]{wainwright2005new}
Wainwright, M.~J.; Jaakkola, T.~S.; and Willsky, A.~S. 2005.
\newblock A new class of upper bounds on the log partition function.
\newblock \emph{IEEE Transactions on Information Theory}, 51(7): 2313--2335.

\bibitem[{Zermelo(1929)}]{zermelo1929berechnung}
Zermelo, E. 1929.
\newblock Die Berechnung der Turnier-Ergebnisse als ein Maximumproblem der
  Wahrscheinlichkeitsrechnung.
\newblock \emph{Mathematische Zeitschrift}, 30: 436--460.

\bibitem[{Élö(1978)}]{elo1978rating}
Élö, A.~E. 1978.
\newblock \emph{The Rating of Chessplayers, Past and Present}.
\newblock New York: Arco Pub.

\end{thebibliography}

\appendix
\newpage
\onecolumn

\begin{center}
     {\huge Appendix}
 \end{center}

\section{Proof of Theorem~\ref{theo:Phi}}
\label{app:Phi}

Theorem~\ref{theo:Phi} is classical\footnote{To quote Wikipedia$\copyright$, ``The cumulant generating function [...], if it exists, is infinitely differentiable and convex, and passes through the origin. Its first derivative ranges monotonically in the open interval from the infimum to the supremum of the support of the probability distribution, and its second derivative is strictly positive everywhere it is defined, except for the degenerate distribution of a single point mass." Source:~\url{https://en.wikipedia.org/wiki/Cumulant}} , and we provide a proof for the sake of completeness.
The moment-generating function $M_f : \theta \mapsto \int_{\R} \mathrm{e}^{\theta r} \mathrm{d}F(r)$ is well-defined over $\R$ since all the exponential moments of $f$ are finite.
    The function $r \mapsto \mathrm{e}^{\theta r}$ is smooth and such that $\partial_{\theta}^n \mathrm{e}^{\theta r} = r^n  \mathrm{e}^{\theta r}$. 
    For any $\theta , \theta' \in \R$ such that $\theta' > \theta$, we have that
    $|r^n \mathrm{e}^{\theta r}| \leq M \mathrm{e}^{|\theta' r|}$ for any $r \in \Dc$ and some $M > 0$. Since $\mathrm{e}^{|\theta' r|} \in L_1(\Dc, p)$, we deduce that $r^n \mathrm{e}^{\theta r} \in L_1(\Dc, p)$. This implies that $M$ is infinitely differentiable with derivatives given by
$  M^{(n)}(\theta) = 
        \int_{\R} r^n \mathrm{e}^{\theta r} \mathrm{d}F(r).$ 
    Moreover, $\int_{\R} \mathrm{e}^{\theta r} \mathrm{d}F(r) > 0$, which implies that $\Phi (\theta) = \log \int_{\R} \mathrm{e}^{\theta r} \mathrm{d}F(r)$ is well-defined and infinitely smooth. 

    Using the relations $( \log g)' = g' / g$ and $(\log g)'' = (g g'' - (g')^2 ) / (g^2)$, we deduce that
    \begin{equation} \label{eq:phiprimeab}
        \Phi_f'(\theta) = \frac{ \int_{\R} r \mathrm{e}^{\theta r} \mathrm{d}F(r)}{ \int_{\R} \mathrm{e}^{\theta r} \mathrm{d}F(r)}.
    \end{equation}
    and 
    \begin{equation}  \label{eq:phiprimeprimeab}
        \Phi_f''(\theta) = \frac{ \left( \int_{\R} r^2 \mathrm{e}^{\theta r} \mathrm{d}F(r) \right)  \left( \int_{\R}  \mathrm{e}^{\theta r} \mathrm{d}F(r) \right)  - \left( \int_{\R} r \mathrm{e}^{\theta r} \mathrm{d}F(r) \right) ^2}{\left( \int_{\R}  \mathrm{e}^{\theta r} \mathrm{d}F(r) \right)^2 }
    \end{equation}   
    Using the Cauchy-Schwarz relation, we then have
        \begin{equation}  \label{eq:cauchysch}
    \left( \int_{\R} r \mathrm{e}^{\theta r} \mathrm{d}F(r) \right) ^2 =
    \left( \int_{\R} r \sqrt{\mathrm{e}^{\theta r}} \times \sqrt{\mathrm{e}^{\theta r}} \mathrm{d}F(r) \right) ^2
    \leq 
\left( \int_{\R} r^2 \mathrm{e}^{\theta r} \mathrm{d}F(r) \right)  \left( \int_{\R}  \mathrm{e}^{\theta r} \mathrm{d}F(r) \right),
        \end{equation}   
        with equality if and only if $r \sqrt{\mathrm{e}^{\theta r}} = \lambda \sqrt{\mathrm{e}^{\theta r}}$ for (almost) any $r \in  \R$ and some $\lambda > 0$. 
        This is only possible if $\R = \{\lambda\}$, which is excluded. This shows that $\left( \int_{\R} r \mathrm{e}^{\theta r} \mathrm{d}F(r) \right) ^2 < \left( \int_{\R} r^2 \mathrm{e}^{\theta r} \mathrm{d}F(r) \right)  \left( \int_{\R}  \mathrm{e}^{\theta r} \mathrm{d}F(r) \right)$, and therefore that $\Phi_f''(\theta) > 0$ according to \eqref{eq:phiprimeprimeab}. This implies that $\Phi$ is strictly convex and that $\Phi'$ is strictly increasing. The latter is therefore a continuous bijection from $\R$ to its image.

        It is moreover known that $\Phi_f'(\theta) = \frac{M'_f(\theta)}{M_f(\theta)} \rightarrow_{\theta \rightarrow \infty} r_{\max}$ for finite $r_{\max}$~\cite[Theorem 2.7.9]{durrett2019probability}, where we recall that $M_f(\theta)$ is the moment generating function of $f$. This proves that $\Phi_f'(\R) = (-r_{\max},r_{\max})$. 
        Finally, knowing that $f$ is symmetric with respect to $0$, by change of variable $r \rightarrow -r$ in \eqref{eq:defphiab}, we see that $\Phi_f(\theta) = \Phi_f(-\theta)$ and $\Phi$ is even. The derivative and second derivative of an even function are odd and even, respectively.

\section{Proofs on MAP Estimators}

This section contains the proofs of Proposition~\ref{prop:existuniquemin} to \ref{prop:boundinfinity}. 

\begin{proof}[Proof of Proposition~\ref{prop:existuniquemin}]
    According to \eqref{eq:be}, the negative log-posterior $\Theta \mapsto \mathcal{L}(\Theta|\Comp)$ and  $\Theta \mapsto \sum_{a \in \mathcal{A}} \frac{ \theta^2_a}{2 \sigma^2} + \sum_{(ab) \in \mathcal{C}} \Phi_f (\theta_{ab} ) - r_{ab} \theta_{ab}$ differ from a constant. The latter is $\alpha$-strongly convex due to the quadratic term,
    hence $\mathcal{L}(\Theta|\Comp)$ is $(1/\sigma^2)$-strongly convex over $\R^A$.
   The existence and uniqueness of $\Theta_{f,\sigma^2}^*(\Comp)$ directly follows.
\end{proof}

\begin{proof}[Proof of Proposition~\ref{prop:bayesianconstraint}]
According to Proposition~\ref{prop:existuniquemin}, there exists a unique minimizer $\Theta_{f,\sigma^2}^* (\Comp) = \arg\min_{\Theta \in \Ds^A} \mathcal{L}(\Theta|\Comp).$
    Then, $\Theta_{f,\sigma^2}^* (\Comp)$ is characterized by the equation $\nabla_{\Theta} \mathcal{L}( \Theta_{f,\sigma^2}^*(\Comp) | \Comp ) = 0$. For any $a$, we therefore have that
    \begin{equation} \label{eq:usefullater}
        0 = \partial_{\theta_a} \mathcal{L} ( \Theta_{f,\sigma^2}^* (\Comp) | \Comp ) 
        =
        \frac{\theta_a^*(\Comp)}{\sigma^2} + \sum_{b \in \mathcal{A}_a} \Phi_f'( \theta_{ab}^* (\Comp)  ) -   r_{ab}. 
    \end{equation}
    Hence, summing over $a$, we have 
    \begin{equation} \label{eq:computesumtheta*}
        0 =  \sum_{a \in \mathcal{A}} \partial_{\theta_a} \mathcal{L} ( \Theta_{f,\sigma^2}^* (\Comp) | \Comp ) 
        =
          \frac{1}{\sigma^2} \sum_{a \in \mathcal{A}}  \theta_a^*(\Comp)+ \sum_{(ab) \in \mathcal{C}} \Phi_f'( \theta_{ab}^* (\Comp)  ) -  r_{ab}. 
    \end{equation}    
    Since $r_{ab} = - r_{ba}$ and $\Phi_f'$ is odd (hence $\Phi_f'(\theta_{ab}) = - \Phi_f'(\theta_{ba})$), we have that $ \sum_{(ab) \in \mathcal{C}} r_{ab} = 0 = \sum_{(ab) \in \mathcal{C}} \Phi_f'( \theta_{ab}^* (\Comp)  )$.
    Hence, \eqref{eq:computesumtheta*} leads to $0 =   \sum_{a \in \mathcal{A}}  \theta_a^*(\Comp)$ as expected.
\end{proof}

\begin{proof}[Proof of Proposition~\ref{prop:momentstheta}]
    We fix $(ab) \in \mathcal{C}$.
    We have that
    \begin{equation}
        \mathbb{E}[r_{ab}|\Theta] = \int_{Dc} r \mathrm{d} p_{r_{ab}|\Theta}(r) = \frac{ \int_{\Dc} r \mathrm{e}^{\theta_{ab} r} \mathrm{d}p(r)}{ \int_{\Dc} \mathrm{e}^{\theta_{ab} r} \mathrm{d}p(r)} = \Phi_f'(\theta_{ab}),
    \end{equation}
    where the second equality is identical to \eqref{eq:cecicela} and the third one has been proved in \eqref{eq:phiprimeab}. 
    We also have that, similarly,
    \begin{align}
        \mathbb{V}[r_{ab}|\Theta] &= \mathbb{E} [r^2_{ab}|\Theta] - \mathbb{E}[r_{ab}|\Theta]^2
        = \frac{ \int_{\Dc} r^2 \mathrm{e}^{\theta_{ab} r} \mathrm{d}p(r)}{ \int_{\Dc} \mathrm{e}^{\theta_{ab} r} \mathrm{d}p(r)} - \frac{ \left( \int_{\Dc} r \mathrm{e}^{\theta_{ab} r} \mathrm{d}p(r) \right)^2}{ \left( \int_{\Dc} \mathrm{e}^{\theta_{ab} r} \mathrm{d}p(r) \right)^2} \\
        &= \frac{ \left( \int_{\Dc} r^2 \mathrm{e}^{\theta r} \mathrm{d}p(r) \right)  \left( \int_{\Dc}  \mathrm{e}^{\theta_{ab}  r} \mathrm{d}p(r) \right)  - \left( \int_{\Dc} r \mathrm{e}^{\theta_{ab}  r} \mathrm{d}p(r) \right) ^2}{\left( \int_{\Dc}  \mathrm{e}^{\theta_{ab}  r} \mathrm{d}p(r) \right)^2 }
        = \Phi_f''(\theta_{ab} ),
    \end{align}
    where the last equality has been proven in \eqref{eq:phiprimeprimeab}.
    Let $(cd) \in \mathcal{C} \backslash \{(ab), (ba)\}$. Then, $r_{ab}$ and $r_{cd}$ are independent conditionally to $\Theta$, hence 
       $ \mathbb{E}[r_{ab} r_{cd} |\Theta] = \mathbb{E}[r_{ab}|\Theta]\mathbb{E}[r_{cd}|\Theta] = \Phi_f'(\theta_{ab})(\Phi_f')_{cd}(\theta_{cd})$
    and therefore 
    \begin{equation}
        \mathbb{C}\mathrm{ov} [r_{ab}, r_{cd} |\Theta] =  \mathbb{E}[r_{ab}r_{cd}|\Theta] -  \mathbb{E}[r_{ab}|\Theta] \mathbb{E}[r_{cd}|\Theta] = 0.
    \end{equation}
\end{proof}

    \begin{proof}[Proof of Proposition~\ref{prop:boundinfinity}]
Using the relation~\eqref{eq:usefullater}, 
we observe that $\theta_a^*(\Comp) \leq \sigma^2 \sum_{b \in \mathcal{A}_a} r_{ab} - \Phi_f'(\theta_{ab}^*(\Comp))$. 
Using that $|r_{ab}|\leq r_{\max}$ and $|\Phi_f'(\theta)| \leq r_{\max}$ for any 
$\theta$ (see Theorem~\ref{theo:Phi}), we deduce that
\begin{equation}
    |\theta_a^*(\Comp)| \leq \sigma^2 \sum_{b \in \mathcal{A}_a} |r_{ab}| + | \Phi_f'(\theta_{ab}^*(\Comp))| \leq 2r_{\max} A_a \sigma^2,
\end{equation}
as expected. 
    \end{proof}

\section{Proofs on Monotonicity}

This section contains the proofs of Proposition~\ref{prop:onpartialrab}, \ref{prop:ondeltarabdiscrete}, and \ref{prop:equalscores}. 
\begin{proof}[Proof of Proposition~\ref{prop:onpartialrab}]
    For $(ab) \in \mathcal{C}$,  we define the elementary comparison matrix
    \begin{equation}
        \mathrm{E}^{ab}= ( \delta_{(cd) = (ab)} - \delta_{(cd) = (ba)} )_{(cd)\in \mathcal{C}}. 
    \end{equation}
    Assume that \eqref{eq:partialtheta} holds. We fix $a \in \mathcal{A}$ and $\Comp$ and $\Comp'$ such that $\Comp \leq_a \Comp'$.
    Let $A_a = |\mathcal{A}_a|$ be the number of elements $b$ compared with $a$. We number these elements $(b_1, \ldots , b_{A_a})$.    
    We define recursively $\Comp_0 = \Comp$ and, assuming $\Comp_{n-1}$ for $1 \leq n \leq A_a$, 
    \begin{equation}
        \Comp_n = \Comp_{n-1} + (r_{a b_n}' - r_{a b_n} ) \mathrm{E}^{ab_n}.
    \end{equation}
    We easily verify that $\Comp_{A_a} = \Comp'$. We therefore have that
    \begin{equation} \label{eq:machin1}
        \widehat{\theta}_a(\Comp') - \widehat{\theta}_a(\Comp) =
        \widehat{\theta}_a(\Comp_{A_a}) - \widehat{\theta}_a(\Comp_0) 
        =\sum_{n=1}^{A_a} ( \widehat{\theta}_a(\Comp_{n}) - \widehat{\theta}_a(\Comp_{n-1}) ).
    \end{equation}
    We observe that $\Comp_{n-1}$ and $\Comp_n$ only differ by $1$ elements indexed by $(ab_n)$, hence
    \begin{equation} \label{eq:machin2}
         \widehat{\theta}_a(\Comp_{n}) - \widehat{\theta}_a(\Comp_{n-1}) = \int_{0}^{r_{ab_n}' - r_{ab_n}} \partial_{r_{ab_n}} \widehat{\theta}_a(\Comp_{n-1} + r \mathrm{E}^{ab_n}) \mathrm{d}r. 
    \end{equation}
    Since $\partial_{r_{ab_n}} \widehat\theta_a(\Comp_{n-1} + r \mathrm{E}^{ab_n}) \leq 0$ by assumption, we deduce that \eqref{eq:machin2} and therefore \eqref{eq:machin1} are non-negative, hence $  \widehat{\theta}_a(\Comp') \leq \widehat{\theta}_a(\Comp)$. This is true for any $a\in \mathcal{A}$, hence $\widehat{\Theta}(\Comp)$ is increasing with respect to $\Comp$ in the sense of Definition~\ref{def:increasingscoring}. \\

    We now assume that $\widehat{\Theta}(\Comp)$ is increasing with respect to $\Comp$. We fix $(ab) \in \mathcal{C}$.     For $h > 0$, $\Comp + h \mathrm{E}^{ab} \geq_a \Comp$, hence $\widehat{\theta}_a(\Comp + h \mathrm{E}^{ab})\geq  \widehat{\theta}_a(\Comp)$ by assumption. This shows that
    \begin{equation}
        \partial_{r_{ab}} \widehat{\theta}_a(\Comp) = \lim_{h\rightarrow 0, h > 0} \frac{\widehat{\theta}_a(\Comp + h \mathrm{E}^{ab}) - \widehat{\theta}_a(\Comp)}{h} \geq 0
    \end{equation}
    as a limit of non negative quantities. This shows \eqref{eq:partialtheta} and concludes the proof. 

    For the strict monotonicity, the proof is the same by considering comparisons $\Comp <_a \Comp'$ and by remarking that $\partial_{r_{ab_n}} \widehat\theta_a(\Comp_{n-1} + r \mathrm{E}^{ab_n}) > 0$ for at least one $n$, ensuring that \eqref{eq:machin2} is strictly positive for at least one $n$, hence $  \widehat{\theta}_a(\Comp') > \widehat{\theta}_a(\Comp)$. 
\end{proof}

\begin{proof}[Proof of Proposition~\ref{prop:ondeltarabdiscrete}]
    The proof is similar to the one of Proposition~\ref{prop:onpartialrab}. 
    The condition $ \Delta_{(ab)} \widehat{\theta}_a(\Comp) \geq 0$ is equivalent to the monotonicity of elementary comparison updates (\emph{i.e.}, we transform $\Comp$ over a single pair $(ab)$ and increase the comparison value by one unit in the discrete comparison domain $\Dc$ for $(ab)$ and reduce it by one unit for $(ba)$). The relation \eqref{eq:partialthetadiscrete} is therefore a direct consequence of the monotonicity of $\widehat{\Theta}(\Comp)$. 
    
    Conversely, assume that $\Comp \leq_a \Comp'$ for any $a \in \mathcal{A}$. For $a \in \mathcal{A}$ fixed, we transform the comparison matrices iteratively using elementary comparison updates such that $\Comp_0 = \Comp$ and $\Comp_N = \Comp'$ for some $N\geq 1$ counting the number of required elementary updates. We also denote by $(a b_n) \in \mathcal{C}$ the comparison pair updated at iteration $n$ and by $r_{k_n}$ the value of $\Comp_n$ at position $(ab_n)$, such that $\Comp_{n} =  (\Comp_{n-1})^{\mathrm{up}}_{(ab_n)}$.  We then have that
    \begin{equation} \label{eq:machin38}
        \widehat{\theta}_a(\Comp') - \widehat{\theta}_a(\Comp) =
        \widehat{\theta}_a(\Comp_{N}) - \widehat{\theta}_a(\Comp_0) 
        =\sum_{n=1}^{N} ( \widehat{\theta}_a(\Comp_{n}) - \widehat{\theta}_a(\Comp_{n-1}) )
        =\sum_{n=1}^{N} \frac{\Delta_{(a b_n)} \widehat{\theta}_a(\Comp_{n-1})}{r_{k_n+1} - r_{k_n}}
         \geq 0,
    \end{equation}
    since $\Delta_{(a b_n)} \widehat{\theta}_a(\Comp_{n-1}) \geq 0$ by assumption for any $n$. This is valid for any $a \in \mathcal{A}$, hence $\widehat{\Theta}(\Comp)$ is increasing with respect to $\Comp$. 
\end{proof}

\begin{proof}[Proof of Proposition~\ref{prop:equalscores}]
    Recall that the non-negative log-posterior is given by $\mathcal{L}(\Theta|\Comp) =  \sum_{c \in \mathcal{A}} \frac{ \theta_{c}^2}{2 \sigma^2} + \sum_{(cd) \in \mathcal{C}} \Phi_f(\theta_{cd}) - r_{cd} \theta_{cd}$. Using that $\Comp$ and $\Comp'$ coincides over $\mathcal{C}'$ and $ \mathcal{L} (\Theta^*(\Comp)|\Comp)=0$, we deduce that
    \begin{equation} \label{eq:bidule}
        \mathcal{L} (\Theta^*(\Comp)|\Comp') =  \mathcal{L} (\Theta^*(\Comp)|\Comp') - \mathcal{L} (\Theta^*(\Comp)|\Comp) =  \Phi_f(\theta_{ab}^*(\Comp)) - r_{ab}'\theta_{ab}^*(\Comp) . 
    \end{equation}
    We observe that $ \Theta^*(\Comp) = \Theta^*(\Comp')$ if and only if $\nabla  \mathcal{L} (\Theta^*(\Comp) |\Comp') = 0$ if and only if $\Phi_f'(\theta_{ab}^*(\Comp)) = r_{ab}'$ according to \eqref{eq:bidule}, as expected.

    Finally, the strict monotonicity of $r'_{ab}\mapsto \Theta^*(\Comp')$ (Theorem~\ref{theo:monotone}) implies that $\Theta^*(\Comp) <  \Theta^*(\Comp')$ for $r_{ab}' < \Phi_f'(\theta_{ab}^*(\Comp))$ and $\Theta^*(\Comp) >  \Theta^*(\Comp')$ for $r_{ab}' > \Phi_f'(\theta_{ab}^*(\Comp))$.
\end{proof}

\section{Proof of Theorem~\ref{theo:monotone}}

    We first provide a short recap and a useful result on diagonally dominant matrices, which will be used in the proof of Theorem~\ref{theo:monotone}.

    \begin{definition} \label{def:diago}
        We say that a matrix $\mathrm{M} = (m_{ab})_{(ab) \in \mathcal{A}^2} \in \R^{A \times A}$ is \emph{diagonally dominant} if 
        \begin{equation}
            |m_{aa}| \geq \sum_{b \in \mathcal{A}_a} |m_{ab}|
        \end{equation}
        for any $a \in \mathcal{A}$.
        It is moreover \emph{strictly diagonal dominant} if the inequalities are strict.
    \end{definition}

    \begin{lemma} 
         \label{prop:invertingthesystem}
        Let $\mathrm{M} = (m_{ab})_{(ab) \in \mathcal{A}^2}$ be a symmetric and strictly diagonally dominant matrix such that  $m_{aa} > 0$ for any $a \in \mathcal{A}$ and $m_{ab} \leq 0$ for any $a \neq b$.
        Then, $\mathrm{M}$ is invertible, positive-definite, and its inverse $\mathrm{M}^{-1} = \mathrm{N} = (n_{ab})_{(ab) \in \mathcal{A}^2}$ is such that 
        \begin{equation}
            n_{aa} > n_{ab} \geq 0
        \end{equation}
        for any $(ab) \in \mathcal{A}^2$.
    \end{lemma}

    \begin{proof}
    The matrix $\mathrm{M}$ is strictly diagonally dominant, it is therefore invertible~\cite[Theorem 6.1.10 (a)]{horn2012matrix}. Being symmetric with positive diagonal entries, it is moreover positive definite~\cite[Theorem 6.1.10 (c)]{horn2012matrix}.
    The matrix $\mathrm{M}$ is therefore a M-matrix, which is known to be equivalent to the fact that its inverse $\mathrm{N}$ has non-negative entries \cite[Definition 1.1]{pena1995m}, hence $n_{ab} \geq 0$ for any $(ab)$.

    We fix $(ab) \in \mathcal{A}^2$. Without loss of generality, one assumes that $a$ is the first index in $\mathrm{M}$. 
    We partition $\mathrm{M}$ as
    \begin{equation}
      \mathrm{M} =  \begin{bmatrix}
m_{aa} & - v^T\\
- v & \mathrm{M}'
\end{bmatrix}
    \end{equation}
    where $v \in \R^{A-1}$ have non-negative entries and where $\mathrm{M}' \in \R^{(A-1)\times(A-1)}$ is also a strictly diagonally dominant  $M$-matrix. As seen below, this implies that $\mathrm{M}'$ is invertible with inverse $\mathrm{N}'$ having non-negative entries. 
    Using Schur complement, we deduce that 
    \begin{equation}
      \mathrm{N} = (n_{cd})_{(cd) \in \mathcal{A}^2} = \mathrm{M}^{-1} =  \begin{bmatrix}
(m_{aa} - v^T \mathrm{M}' v)^{-1} & *\\
(m_{aa} - v^T \mathrm{M}' v)^{-1} (\mathrm{M}')^{-1} v & *
\end{bmatrix}.
\end{equation}
Then, we have that
\begin{equation} \label{eq:toworkonit}
    n_{aa} - n_{ab} = (m_{aa} - v^T \mathrm{M}' v)^{-1} ( 1 - ((\mathrm{M}')^{-1} v)_b). 
\end{equation}
We already know that $\mathrm{N}$ has non-negative entries. It is moreover positive-definite since its inverse $\mathrm{M}$, hence $n_{aa}\neq 0$. This implies that $(m_{aa} - v^T \mathrm{M}' v)^{-1} > 0$. 

 We set $e= (1,\ldots ,1)^T$. Then, the vector $u = \mathrm{M}' e - v$ has entries $u_c = \sum_{d \in \mathcal{A}} m_{cd} > 0$ since $\mathrm{M}$ is strictly diagonally dominant. Using that $\mathrm{N}' = (\mathrm{M}')^{-1}$ has strictly positive entries, we deduce that $\mathrm{N}' u = e - \mathrm{N}' v$ has non-negative entries. $e_b - ((\mathrm{M}')^{-1} v)_b) = 1 - ((\mathrm{M}')^{-1} v)_b) > 0$. This implies that $ n_{aa} - n_{ab} > 0$ as a product of strictly positive terms, as expected.
    \end{proof}

\begin{proof}[Proof of Theorem~\ref{theo:monotone}]
We first treat the case when the domain of $\Comp$ is a continuum $\Dc$. Under the assumptions of the GBT model, $\Theta^*(\Comp)$ is differentiable with respect to the comparisons. 
According to Proposition~\ref{prop:onpartialrab}, it suffices to show that $\partial_{r_{ab}} \theta_a^* (\Comp) > 0$
for  any $(ab) \in \mathcal{C}$  and  any $\Comp \in \Dc^C$.
We fix $(ab) \in \mathcal{C}$. We have that $\nabla \mathcal{L} (\Theta^*(\Comp)|\Comp) = 0$, hence, we have for any $c$ that
    \begin{equation} \label{eq:truc1}
        0 =  \frac{\theta_c^* (\Comp)}{\sigma^2}   + \sum_{d \in \mathcal{A}_c} \Phi_f'(\theta_{cd}^* (\Comp)) - \sum_{d \in \mathcal{A}_c} r_{cd}. 
    \end{equation}
     We observe that $ \partial_{r_{ab}} r_{cd} =  1$ if $(cd) = (ab)$, $-1$ if $(cd) = (ba)$ (since $r_{ba}=-r_{ab}$),  and $0$ otherwise. We therefore deduce that
    \begin{equation}
        \partial_{r_{ab}} \sum_{d \in \mathcal{A}_c} r_{cd} = 
    \delta_{c=a} 1_{b\in \mathcal{A}_a} -\delta_{c=b} 1_{a\in \mathcal{A}_b} 
        =
        (\delta_{c=a} - \delta_{c=b}) 1_{(ab)\in \mathcal{C}} = \delta_{c=a} - \delta_{c=b}, 
    \end{equation}
    using that $(ab) \in \mathcal{C}$. 
    In particular, taking the partial derivative $\partial_{r_{ab}}$ of \eqref{eq:truc1}, we have for any $c$ that
    \begin{equation} \label{eq:truc2}
         \frac{\partial_{r_{ab}} \theta_c^* (\Comp)}{\sigma^2} + \sum_{d \in \mathcal{A}_c}  \Phi_f'' (  \theta_{cd}^* (\Comp) )  \partial_{r_{ab}}  \theta_{cd}^* (\Comp)  = \delta_{c=a} - \delta_{c = b} . 
    \end{equation}
    We define the vectors and matrix
    \begin{align} \label{eq:xcandco}
           x &= (x_c)_{c \in \mathcal{C}} =   \partial_{r_{ab}} \Theta^* (\Comp)  ,
        \quad  
       \lambda = (\lambda_c )_{c \in \mathcal{C}} =  \left(  \frac{1}{\sigma^2} +  \sum_{d  \in \mathcal{A}_c}  \Phi_f'' (  \theta_{cd}^* (\Comp) ) \right)_{c \in \mathcal{C}},  \quad \text{and} \\
    \mathrm{P} &= \left(  \Phi_f'' (  \theta_{cd}^* (\Comp) )  \right)_{(cd) \in \mathcal{C}}.
    \end{align}
    Then, \eqref{eq:truc2} becomes, for any $c$,
    \begin{equation}
        \lambda_c x_c - (\mathrm{P} x)_c =  \delta_{c=a} - \delta_{c = b}  ,
    \end{equation}
    or equivalently
    \begin{equation}
        \mathrm{M} x = e^a - e^b
    \end{equation}
    where $\mathrm{M} = (m_{cd})_{(cd)\in \mathcal{A}^2}$ is such that
    $m_{cc} = \lambda_c$ and $m_{cd} =- p_{cd}$ for any $(cd) \in \mathcal{C}$, $m_{cd}=0$ if $c\neq d$ and $(cd) \notin \mathcal{C}$, and where $e^a = (e^a_b)_{b \in \mathcal{A}} = (\delta_{a=b})_{b \in \mathcal{A}}$ is the $a$th canonical vector. Note that the matrix $\mathrm{M}$  depends on $\Theta^* (\Comp)$ and $\sigma^2$, but not on $(ab) \in \mathcal{C}$.\\

    The matrix $\mathrm{M}$ is symmetric with strictly positive diagonal and negative non-diagonal by definition. It is strictly diagonally dominant due to the relation
    \begin{equation}
        |m_{aa}| = m_{aa} = \lambda_a = \frac{1}{\sigma^2} +  \sum_{b \in \mathcal{A}_a} p_{ab} > \sum_{b \in \mathcal{A}_a}p_{ab} =
        \sum_{b \in \mathcal{A}_a} |m_{ab}|, 
    \end{equation}
    where we used that $\sigma^2 > 0$  and that $p_{ab} = - m_{ab} = |m_{ab}|$ for any $(ab) \in \mathcal{C}$.
    The matrix $\mathrm{M}$ hence satisfies the conditions of Proposition~\ref{prop:invertingthesystem} above. It is therefore invertible and its inverse $\mathrm{N}= (n_{ab})_{(ab)\in \mathcal{C}}$ satisfies $n_{aa} > n_{ab}$ for any $(ab) \in \mathcal{C}$.   
    We deduce that, for any $\Comp$,
    \begin{equation}
        \delta_{r_{ab}}\theta_a^* (\Comp)  = x_a = \left( \mathrm{N}  (e^a - e^b) \right)_a =   n_{aa} - n_{ab}  > 0, 
    \end{equation}  
    as expected. \\

    In the case where the domain $\Dc$ is discrete, we follow a similar proof. The main difference is that we rely on the discrete monotonicity criteria given in Proposition~\ref{prop:ondeltarabdiscrete}.  For $(ab) \in \mathcal{C}$, we apply the finite difference operator $\Delta_{(ab)}$ defined in \eqref{eq:defdeltaab} to \eqref{eq:truc1}. We have that
    \begin{equation}
        \Delta_{(ab)} \Phi_f'(\theta^*_{cd}(\Comp)) = \frac{\Phi_f'(\theta^*_{cd}(\Comp^{\mathrm{up}}_{ab}))- \Phi_f'(\theta^*_{cd}(\Comp))}{r_{k+1}-r_k}
        = \Phi_f'' ( \theta_{cd}^{0}(\Comp)) \Delta_{(ab)} \theta_{cd}^*(\Comp) 
    \end{equation}    
    for some $ \theta_{cd}^{0}(\Comp) \in \R$ (Taylor-Lagrange formula). We deduce therefore that $\Delta_{(ab)}$ applied to \eqref{eq:truc1} gives
        \begin{equation} \label{eq:truc19}
         \frac{\Delta_{(ab)} \theta_c^* (\Comp)}{\sigma^2} + \sum_{d \in \mathcal{A}_c}  \Phi_f'' (  \theta_{cd}^{0}(\Comp)) \Delta_{(ab)}  \left( \theta_{cd}^* (\Comp)  \right)  = \sum_{d \in \mathcal{A}_c} \Delta_{(ab)} r_{cd} = \delta_{c=a} - \delta_{c = b} . 
    \end{equation}
    Setting $x_c = \Delta_{(ab)} \theta_c^*(\Comp)$, we then follow an identical proof as in the continuous case to deduce that $\Delta_{(ab)}\theta_a^*(\Comp) > 0$. 
    \end{proof}
\section{Proof of Theorem~\ref{theo:resilient}}

The following classical Lemma will be used in the proof of Theorem~\ref{theo:resilient}. 

    \begin{lemma}
\label{lemma:minimum_deviation_strongly_convex}
Let $f : \setR^d \rightarrow \setR$ an $\alpha$-strongly convex function and denote $x^*$ its minimum.
Then, for any $x \in \setR^d$, 
we have 
\begin{equation}
    \norm{x - x^*}{2} \leq \frac{2}{\alpha} \norm{\nabla f(x)}{2}. 
\end{equation}
\end{lemma}

\begin{proof}
First recall that strongly convex functions always have a unique minimum.
Now let $x \in \setR^d$.
By $\alpha$-strong convexity, we know that
\begin{align}
    f(x^*) 
    &\geq f(x) + \nabla f(x)^T (x^* - x) + \frac{\alpha}{2} \norm{x^* - x}{2}^2 \\
    &\geq f(x^*) - \norm{\nabla f(x)}{2} \norm{x^* - x}{2} + \frac{\alpha}{2} \norm{x^* - x}{2}^2.
\end{align}
Therefore $\norm{\nabla f(x)}{2} \geq \frac{\alpha}{2} \norm{x^* - x}{2}$.
Rearranging terms yields the lemma.
\end{proof}

    \begin{proof}[Proof of Theorem \ref{theo:resilient}]
   The function $\Theta \mapsto \mathcal{L} (\Theta|\Comp)$ is $(1/\sigma^2)$-strongly convex. Applying Lemma~\ref{lemma:minimum_deviation_strongly_convex} to $f = \mathcal{L}(\cdot |\Comp)$, for which $x^* = \Theta_{f,\sigma^2}^*(\Comp)$, and to $x = \Theta_{f,\sigma^2}^*(\Comp')$, we deduce that
   \begin{equation} \label{eq:boundwithgrad}
       \| \Theta_{f,\sigma^2}^*(\Comp) - \Theta_{f,\sigma^2}^*(\Comp') \|_2 \leq 2 \sigma^2 \| \nabla \mathcal{L} ( \Theta_{f,\sigma^2}^*(\Comp') |\Comp) \|_2.
   \end{equation}
    We fix $a \in \mathcal{A}$ and assume that $\Comp$ and $\Comp'$ only differ from one elementary modification, \emph{i.e.} we can transform $\Comp$ into $\Comp'$ by adding, removing, or updating a comparison. This means in particular that $\Delta(\Comp, \Comp') = 1$. 
    As we have seen in \eqref{eq:usefullater}, we have that
    \begin{equation}
        \partial_{\theta_a} \mathcal{L} ( \Theta_{f,\sigma^2}^*(\Comp') |\Comp) = 
                 \frac{\theta_a^* (\Comp')}{\sigma^2}  + \sum_{b \in \mathcal{A}_a} \Phi_f'( \theta_{ab}^* (\Comp')  ) -   r_{ab}. 
    \end{equation}
    Using that $0 = \partial_{\theta_a} \mathcal{L} ( \Theta_{f,\sigma^2}^* (\Comp') | \Comp' )  = \frac{\theta_a^* (\Comp' )}{\sigma^2} + \sum_{b \in \mathcal{A}_a'} \Phi_f'( \theta_{ab}^* (\Comp')  ) - \sum_{b \in \mathcal{A}_a'} r_{ab}'$, we deduce that
    \begin{align} \label{eq:partialstuff}
       \partial_{\theta_a} \mathcal{L} ( \Theta_{f,\sigma^2}^*(\Comp') |\Comp)  
       = \sum_{b \in \mathcal{A}_a \backslash \mathcal{A}_a'} (\Phi_f'(\theta_{ab}^*(\Comp'))- r_{ab} ) 
        -
        \sum_{b \in \mathcal{A}_a' \backslash \mathcal{A}_a} ( \Phi_f'(\theta_{ab}^*(\Comp'))- r'_{ab}  )
        +
        \sum_{b\in \mathcal{A}_a\cap \mathcal{A}_a'} ( r'_{ab} - r_{ab}). 
    \end{align}
    Using that $\sup_{\Ds} | \Phi_f'| \leq r_{\max}$ (see Theorem~\ref{theo:Phi}), we deduce from \eqref{eq:partialstuff} that 
    \begin{equation}
       |\partial_{\theta_a} \mathcal{L} ( \Theta_{f,\sigma^2}^*(\Comp') |\Comp)  | \leq  2r_{\max} ( | \mathcal{A}_a \backslash \mathcal{A}_a'| + | \mathcal{A}_a' \backslash \mathcal{A}_a|
       + | \mathcal{A}_a \cap \mathcal{A}_a'|).
    \end{equation}
    Since $\Comp$ and $\Comp'$ only differ from one comparison, we have that $ | \mathcal{A}_a \backslash \mathcal{A}_a'| + | \mathcal{A}_a' \backslash \mathcal{A}_a|
       + | \mathcal{A}_a \cap \mathcal{A}_a'| = 1$ if $a \in \{a_n, b_n \}$ and $0$ otherwise.        
       This shows that  
       \begin{align} \label{eq:boundthisgrad}
           \| \nabla \mathcal{L} ( \Theta_{f,\sigma^2}^*(\Comp') |\Comp)  \|_2 &= \sqrt{\sum_{a \in \mathcal{A}}  |\partial_{\theta_a} \mathcal{L} ( \Theta_{f,\sigma^2}^*(\Comp') |\Comp)  |^2} =
           \sqrt{  |\partial_{\theta_{a_n}} \mathcal{L} ( \Theta_{f,\sigma^2}^*(\Comp') |\Comp)  |^2 + |\partial_{\theta_{b_n}} \mathcal{L} ( \Theta_{f,\sigma^2}^*(\Comp') |\Comp)  |^2} \nonumber \\
           &\leq \sqrt{8 (r_{\max})^2}= 2 \sqrt{2} r_{\max}
       \end{align}
       as soon as $\Delta(\Comp, \Comp') = 1$. 
    
    We define a list of comparison matrices $\Comp^n = (r_{ab}^n)_{(ab) \in \mathcal{C}^n}$ such that $\Comp^0 = \Comp$ and $\Comp^{n}$ is obtained from $\Comp^{n-1}$ by adding to $\mathcal{C}^n$, removing from $\mathcal{C}^{n-1}$, or changing in $\Comp^{n-1}$ one comparison $a_{n}b_{n}$ of $\Comp$ for one of $\Comp'$. We have that $\Comp^N = \Comp'$ with $N = \Delta( \Comp, \Comp')$. We then have that
    \begin{align}
        \| \Theta_{f,\sigma^2}^*(\Comp) - \Theta_{f,\sigma^2}^*(\Comp') \|_2
        &= 
        \left\|  \sum_{n=1}^{\Delta( \Comp, \Comp')}  \Theta_{f,\sigma^2}^*(\Comp^n) - \Theta_{f,\sigma^2}^*(\Comp^{n-1})  \right\|_2 
        \leq 
        \sum_{n=1}^{\Delta( \Comp, \Comp')} \| \Theta_{f,\sigma^2}^*(\Comp^n) - \Theta_{f,\sigma^2}^*(\Comp^{n-1}) \|_2 \nonumber \\
        & \leq
        2 \sigma^2 \sum_{n=1}^{\Delta( \Comp, \Comp')}  \| \nabla \mathcal{L} ( \Theta_{f,\sigma^2}^*(\Comp^{n}) |\Comp^{n-1})  \|_2 
        \leq 
        4 \sqrt{2} r_{\max} \sigma^2  \Delta( \Comp, \Comp')
    \end{align}
    where we used \eqref{eq:boundwithgrad} and \eqref{eq:boundthisgrad} for the penultimate and last inequalities, respectively.  This shows that the left quantity in \eqref{eq:boundresilience} is bounded by $4 \sqrt{2} r_{\max}\sigma^2$. 
    \end{proof}

\section{Examples}

       \subsection{$K$-nary GBT Models}
    \label{sec:knary}

We have seen  that the Bernoulli-BT model corresponds to the GBT model for $f(r) = \frac{\delta_{1}+\delta_{-1}}{2}$.
We generalize this model for comparisons over $K$ values uniformly spread over $[-1,1]$. We set for $1 \leq k \leq K$
\begin{equation} \label{eq:defrkK}
    r_k^K = \frac{2(k-1)}{K-1} - 1.
\end{equation}
We have $r_{1}^K = -1$, $r_{K}^K = 1$, and $r_{k+1} - r_k = \frac{2}{K-1}$. 
We observe that $0 \in \Dc$ if and only if $K$ is odd, in which case $r_{(K+1)/2}^K = 0$.
The $K$-nary GBT model corresponds to discrete comparisons with common probability mass function the uniform one, \emph{i.e.}, 
\begin{equation} \label{eq:pmfunif}
   f(r) = \frac{1}{K}
\end{equation}
for any $r \in \Dc= \{r_k^K\}_{1\leq k \leq K}$.

\begin{proposition}
Assume that $\Comp|\Theta$ follows the $K$-nary GBT model with $K\geq 2$. 
We then have that
\begin{equation}
    \Phi_{K}(\theta) = \log\left( \frac{\sinh \left( \frac{K\theta}{K-1}\right) }{K \sinh  \left(\frac{\theta}{K-1} \right)}\right)
\end{equation}
and the probability mass function of $\Comp|\Theta$ is 
\begin{equation} \label{eq:Pcomprkab}
    P \left( \Comp |\Theta \right)
    =
   K^C \prod_{(ab) \in \mathcal{C}} \frac{ \sinh \left( \frac{\theta}{K-1}\right)}{\sinh \left( \frac{K\theta}{K-1}\right) }  \exp \left(  r_{ab} \theta_{ab} \right). 
\end{equation}
\end{proposition}

\begin{proof}
    We have that, using \eqref{eq:Pcompthetadiscrete} and \eqref{eq:pmfunif}, 
    \begin{equation}
      P(\Comp|\Theta) = \prod_{(ab)\in \mathcal{C}} \frac{\mathrm{e}^{r_{ab}\theta_{ab}}}{ \sum_{r \in \Dc} \mathrm{e}^{r\theta_{ab}}}. 
    \end{equation}  
    We therefore want to evaluate $\sum_{r \in \Dc} \mathrm{e}^{r\theta} = \sum_{k=1}^K \mathrm{e}^{r_k^K \theta}$. 
    Using \eqref{eq:defrkK}, we have that
    \begin{equation}
        \sum_{k=1}^K \mathrm{e}^{r_k^K \theta} = \mathrm{e}^{-\theta} \sum_{k=0}^{K-1} \left( \mathrm{e}^{\frac{2\theta}{K-1}} \right)^k 
        = \mathrm{e}^{-\theta} \frac{1 - \mathrm{e}^{\frac{2\theta K}{K-1}}}{1 - \mathrm{e}^{\frac{2\theta }{K-1}}}
        = \mathrm{e}^{-\theta}\frac{\left(\mathrm{e}^{-\frac{\theta K}{K-1}} - \mathrm{e}^{\frac{\theta K}{K-1}} \right) \mathrm{e}^{\frac{\theta K}{K-1}}}{\left(\mathrm{e}^{-\frac{\theta}{K-1}} - \mathrm{e}^{\frac{\theta}{K-1}} \right) \mathrm{e}^{\frac{\theta}{K-1}}} = 
        \frac{\mathrm{e}^{\frac{\theta K}{K-1}} - \mathrm{e}^{- \frac{\theta K}{K-1}}}{\mathrm{e}^{\frac{\theta}{K-1}} - \mathrm{e}^{-\frac{\theta}{K-1}}}. 
    \end{equation}
    This proves that $\Phi_K(\theta) = \log ( \frac{1}{K} \sum_{r\in \Dc}\mathrm{e}^{r\theta}) =
    \log \left(  \frac{\sinh \left( \frac{K\theta}{K-1}\right) }{K \sinh \left( \frac{\theta}{K-1} \right) } \right) $ and then \eqref{eq:Pcomprkab}. 
\end{proof}

    \subsection{Poisson-GBT Model}\label{sec:PoissonBT}

    The random variable $X$ follow a Poisson law of parameter $\lambda > 0$ if $X \in \mathbb{N}$ and $P_{\mathrm{Poisson}}(k) = \mathrm{e}^{-\lambda} \frac{\lambda^k}{k!}$ for any $k \in \mathbb{N}$. We consider the symmetric Poisson law, which is given by $Y= \frac{X - X'}{2}$ where $X$ and $X'$ are independent Poisson random variable with common parameter $\lambda$. This corresponds to the probability mass function given for any $k \in \mathbb{Z} = \Dc$ by
    \begin{equation} \label{eq:poissonsym}
        P(k) = \frac{\mathrm{e}^{-\lambda}}{2} \frac{\lambda^{|k|}}{ |k| !} \text{ if } k \neq 0 \text{ and } P(0) = \mathrm{e}^{-\lambda}.
    \end{equation} 
The Poisson-GBT model corresponds to the discrete GBT model with probability mass function \eqref{eq:poissonsym}. 
All the moments of the Poisson law are finite, as is required. We provide the cumulant generating function and non-Lipschitz resilience of the Poisson-GBT model in the next proposition.

\begin{proposition}
    Assume that $\Comp|\Theta$ follow a Poisson-GBT model with parameter $\lambda > 0$. Then, we have that
    \begin{equation}
        \Phi_{\text{Poisson}}(\Theta) = \lambda \cosh(\theta). 
    \end{equation}
    Moreover, the Poisson GBT is not Lipschitz-resilient (i.e. not $\beta$-Lipschitz-resilient for any $\beta > 0$).
\end{proposition}

\begin{proof}
    Let $X$ be a Poisson random variable. Its cumulant function is known to be $\tilde{\Phi}_{\text{Poisson}}(\theta) = \lambda (\mathrm{e}^{\theta} + 1)$, while the cumulant function of $-X$ is the complex conjugate of $\tilde{\Phi}_{\text{Poisson}}(\theta)$ given by $\tilde{\Phi}_{\text{Poisson}}(-\theta) = \lambda (\mathrm{e}^{-\theta} + 1)$. For the symmetric Poisson law, we therefore have that
    \begin{equation}
        \Phi_{\text{Poisson}}(\theta) = \frac{\tilde{\Phi}_{\text{Poisson}}(\theta) + \tilde{\Phi}_{\text{Poisson}}(-\theta)}{2} = \lambda \frac{\mathrm{e}^{\theta} + \mathrm{e}^{-\theta}}{2} = \lambda \cosh(\theta).
    \end{equation}

    For the Lipschitz-resilience, we consider the case of a single comparison $r=r_{ab}$ between two alternatives $a$ and $b$. Then the MAP estimator $\theta^*(r)$ satisfies $\Phi_{\text{Poisson}}'(\theta^*(r)) + \frac{\theta^*(r)}{\sigma^2} = r$ and is therefore, for any $r \in \Z$, given by
    \begin{equation}
        \theta^*(r) = (\Phi_{\text{Poisson}}' +  \cdot / \sigma^2 )^{-1}(r). 
    \end{equation}
    We therefore deduce that $(\theta^*(r) - \theta^*(r')^2 = (\Phi_{\text{Poisson}}' +  \cdot / \sigma^2)^{-1}(r) \underset{|r-r'| \rightarrow \infty}{\longrightarrow} \infty$, hence the estimator is not Lipschitz-resilient.
\end{proof}

    \subsection{Gaussian-GBT Model}\label{sec:GaussBT}
 
The Gaussian function is denoted by 
$g_{\sigma_0^2}(r) = \frac{1}{\sqrt{2\pi} \sigma_0} \mathrm{e}^{- r^2 / 2 \sigma_0^2}$.
The   Gaussian Bradley-Terry  model corresponds to choosing 
$f(r) = g_{\sigma_0^2}(r)$ with $\Dc = \R$. 
 The moment generating function of the Gaussian random variable is $M(\theta) = \int_{\R} \mathrm{e}^{\theta r} g_{\sigma_0^2}(\theta) \mathrm{d}r = \mathrm{e}^{\sigma_0^2 \theta^2 / 2}$, hence the
 cumulant-generating function of  the
Gaussian-GBT model  is
\begin{equation}
    \Phi_{\text{Gaussian}}(\theta) = \frac{ \sigma_0^2 \theta^2}{2}.
\end{equation}

    In the Gaussian-GBT model with a Gaussian prior, scores and comparisons are Gaussian random variables. We provide the conditional law of the comparisons conditionally to the scores in Proposition~\ref{prop:condgaussian}.
    
\begin{proposition} \label{prop:condgaussian}
    Under the Gaussian-GBT model with a Gaussian prior, we have that, conditionally to $\Theta \in \R^A$, 
\begin{equation}
    p(\Comp |\Theta) 
        =         
         \left( \frac{1}{\sqrt{2\pi} \sigma_0}   \right)^C \exp \left( - \sum_{(ab) \in \mathcal{C}}  \frac{(r - \sigma_0 \theta_{ab})^2}{2 \sigma_0^2}  \right) 
\end{equation}
    In other words,
    $\Comp | \Theta \sim \mathcal{N} \left( (\sigma_0 \theta_{ab})_{(ab)\in \mathcal{C}} , \sigma_0^2 \mathrm{Id} \right)
        $
    is a Gaussian random comparison matrix with independent entries $r_{ab}$ with mean $\theta_{ab}$ and variance $\sigma_0^2$. 
\end{proposition}

\begin{proof}
    According to \eqref{eq:comptheta}, we have that
    \begin{align}
        p(\Comp |\Theta) 
        &= 
        \prod_{(ab)\in \mathcal{C}}  \frac{g_{\sigma_0^2}\left( r_{ab} \right)   
    \mathrm{e}^{r_{ab} \theta_{ab}}}{\mathrm{e}^{\sigma_0^2 \theta_{ab}^2/2}}
 =  \left( \frac{1}{\sqrt{2\pi} \sigma_0}   \right)^C \exp \left( - \sum_{(ab) \in \mathcal{C}} 
 \left( \frac{r_{ab}^2}{2\sigma_0^2} +  \frac{\sigma_0^2 \theta_{ab}^2}{2} - r_{ab} \theta_{ab} \right) \right) \\
        &= \left( \frac{1}{\sqrt{2\pi} \sigma_0}   \right)^C \exp \left( - \sum_{(ab) \in \mathcal{C}}  \frac{(r_{ab} - \sigma_0 \theta_{ab})^2}{2 \sigma_0^2}  \right) .
    \end{align}
\end{proof}

 \paragraph{MAP estimator of the Gaussian-GBT Model.}
 
The MAP estimator of the Gaussian-GBT model with a i.i.d. Gaussian prior is 
\begin{equation} \label{eq:defthetaGBTgauss}
    \Theta_{\text{Gaussian},\sigma^2}^* (\Comp) =  
    \argmin_{\Theta \in \R^{A}} \left\{ \frac{1}{2 \sigma^2} \sum_{a \in \mathcal{A}}  \theta_a^2  + \frac{\sigma_0^2}{2} \sum_{(ab) \in \mathcal{C}}  \theta_{ab}^2 - r_{ab} \theta_{ab} \right\}.
\end{equation}
This leads to a quadratic optimization problem whose solution is linear with respect to the comparison data $\Comp$. In order to solve it, we introduce some notation. The comparison set $\mathcal{C}$ being fixed, we define the matrix
\begin{equation} \label{eq:matAC}
    \mathrm{A}_{\mathcal{C}} = (1_{(ab)\in \mathcal{C} })_{(ab)\in \mathcal{A}^2},
\end{equation}
which is the adjacency matrix associated to the graph induced by $\mathcal{C}$ over $\mathcal{A}$ where $a$ and $b$ are connected if $(ab) \in \mathcal{C}$. 
We  recall that $\mathcal{A}_a$ is the set of $b \in \mathcal{A}$ such that $(ab) \in \mathcal{C}$ and that its cardinal is $A_a$.
We define the diagonal matrix 
\begin{equation} \label{eq:matDaC}
    \mathrm{D}_{\mathcal{C},\sigma^2,\sigma_0^2} = ((1/\sigma^2 + \sigma_0^2 A_a) \delta_{a=b})_{(ab)\in \mathcal{A}^2}.
\end{equation}
We moreover define the vector
\begin{equation} \label{eq:barr}
    \bar{r} = \left( \sum_{b \in \mathcal{A}_a} r_{ab}   \right)_{a \in \mathcal{A} }\in \R^A 
\end{equation}
obtained by summing over each row of the comparison matrix. 

\begin{proposition} \label{prop:gausssolution} 
    The minimizer of \eqref{eq:defthetaGBTgauss} is given by
    \begin{equation} \label{eq:thetaRgauss}
        \Theta_{\text{Gaussian},\sigma^2}^*(\Comp) = \left( \mathrm{D}_{\mathcal{C},\sigma^2,\sigma_0^2} - \sigma_0^2 \mathrm{A}_{\mathcal{C}} \right)^{-1} \bar{r} .
    \end{equation}
       If the user provides the full comparison matrix $\Comp$ (\emph{i.e.} if $\mathcal{C} = \mathcal{A}^2 \backslash \{(a,a), a \in \mathcal{A}\}$), then we have
    \begin{equation}
        \Theta_{\text{Gaussian},\sigma^2}^*(\Comp) = \frac{\bar{r}}{1/\sigma^2+ \sigma_0^2 A} = \frac{\sigma^2}{1 +  \sigma_0^2 \sigma^2 A} \left( \sum_{b \in \mathcal{A} \backslash\{a\}} r_{ab}   \right)_{a \in \mathcal{A} }.
    \end{equation}
\end{proposition}

\begin{proof}
    The vector $\Theta_{\text{Gaussian},\sigma^2}^*(\Comp)$ minimizes the gradient of the loss function of \eqref{eq:defthetaGBTgauss}. This means that, for any $a \in \mathcal{A}$, $0 = \frac{\theta_a^*(\Comp)}{\sigma^2} + \sigma_0^2 \sum_{b \in \mathcal{A}_a} ( \theta_a^*(\Comp) - \theta_b^*(\Comp) - r_{ab} )$. 
    This relation can be rewritten as
    \begin{equation}
        (1/ \sigma^2 + \sigma_0^2 A_a ) \theta_a^*(\Comp) -  \sigma_0^2 \sum_{b \in \mathcal{A}_a}  \theta_b^*(\Comp) =  \sum_{b \in \mathcal{A}_a} r_{ab},
    \end{equation}
    or equivalently, using the notations \eqref{eq:matAC} to \eqref{eq:barr}, 
    \begin{equation}
        (\mathrm{D}_{\mathcal{C},\sigma^2,\sigma_0^2} - \sigma_0^2 \mathrm{A}_{\mathcal{C}}) \Theta_{\text{Gaussian},\sigma^2}^* (\Comp) = \bar{r}.
    \end{equation}
    The matrix $(\mathrm{D}_{\mathcal{C},\sigma^2,\sigma_0^2} - \sigma_0^2 \mathrm{A}_{\mathcal{C}})$ is strictly diagonally dominant (see Definition \ref{def:diago}) and is therefore invertible, proving \eqref{eq:thetaRgauss}. 

    For full comparisons, we have that $A_a = A-1$ for any $a$, hence
    $ \mathrm{D}_{\mathcal{C},\sigma^2,\sigma_0^2} = (1/\sigma^2 + \sigma_0^2 (A-1)) \mathrm{Id}$
    and $\mathrm{A}_{\mathcal{C}} = \mathrm{J}- \mathrm{Id}$ where $\mathrm{J}$ is the matrix with entries $1$. Hence, 
    \begin{equation}
     \mathrm{D}_{\mathcal{C},\sigma^2,\sigma_0^2} - \sigma_0^2 \mathrm{A}_{\mathcal{C}}  = (1/\sigma^2 + \sigma_0^2 (A-1)) \mathrm{Id} -  \sigma_0^2 (\mathrm{J} - \mathrm{I} ) = (1 / \sigma^2 +  \sigma_0^2 A) \mathrm{Id} - \sigma_0^2 \mathrm{J}. 
    \end{equation}
    Using that $\mathrm{J}^2 = A \mathrm{J}$, we have that 
    \begin{equation}
        ((1/\sigma^2 + \sigma_0^2 A) \mathrm{Id} - \sigma_0^2 \mathrm{J})(\mathrm{Id} +  x \mathrm{J}) = (1/\sigma^2+ \sigma_0^2 A) \mathrm{Id} + ( x( 1/\sigma^2 + \sigma_0^2 A) - \sigma_0^2  - \sigma_0^2 Ax )  \mathrm{J} = (1/\sigma^2+ \sigma_0^2 A) \mathrm{Id} + (1/\sigma^2 x - \sigma_0^2) \mathrm{J}.
    \end{equation}
    The previous relation for $x=\sigma_0^2 \sigma^2$ (such that$( x / \sigma^2-\sigma_0^2)=0$) implies that $((1/\sigma^2 + \sigma_0^2A) \mathrm{Id} - \mathrm{J})^{-1} = \frac{\sigma^2}{1 + \sigma_0^2 \sigma^2 A} (\mathrm{Id} + \sigma_0^2 \sigma^2 \mathrm{J})$.
    We also remark that 
    \begin{equation} 
        (\mathrm{J} \bar{r})_a =
    \sum_{b \in \mathcal{A}} \bar{r}_{b} = \sum_{b \in \mathcal{A}} \sum_{c \in \mathcal{A}\backslash \{b\}} r_{bc} = 0 
    \end{equation} 
   for any $a$, hence $\mathrm{J}\bar{r} = 0$. Finally, we obtain that
    \begin{equation}
         \Theta_{\text{Gaussian},\sigma^2}^*(\Comp) =
         \frac{\sigma^2}{1  +  \sigma_0^2 \sigma^2 A} \left( \mathrm{Id} + \frac{\mathrm{J}}{ \sigma_0^2 A} \right) \bar{r}
         = \frac{\bar{r}}{1/\sigma^2+ \sigma_0^2 A},
    \end{equation}
    as expected.
\end{proof}

\paragraph{Monotonicity of the Gaussian-GBT Model.}

According to Theorem~\ref{theo:monotone}, the MAP estimator $\Theta_{\text{Gaussian},\sigma^2}^*(\Comp)$ is strictly increasing with respect to $\Comp$. The proof of Theorem~\ref{theo:monotone} uses an implicit characterization of the minimizer, while Proposition~\ref{prop:gausssolution} provides an explicit formula. We provide an alternative proof of the monotonicity for the Gaussian case using this explicit expression. 

\begin{proof}[Alternative proof of Theorem~\ref{theo:monotone} for Gaussian-GBT]
    We denote $\mathrm{M} = \left( \mathrm{D}_{,\mathcal{C},\sigma_0^2} - \sigma_0^2 \mathrm{A}_{\mathcal{C}} \right)$ and $\mathrm{N} = \mathrm{M}^{-1}$.
    Then, the matrix $M$ is strictly diagonally dominant (see Definition \ref{def:diago}), symmetric, with strictly positive diagonal and non-negative non-diagonal entries. It therefore satisfies the condition of Proposition~\ref{prop:invertingthesystem}. We deduce that the entries $n_{ab}$ of $\mathrm{N}$ satisfy $n_{ab}\geq 0$. 
    
     We fix $(ab) \in \mathcal{C}$. The estimator $\Theta_{\text{Gaussian},\sigma^2}^*(\Comp)$ is differentiable with respect to $\Comp$ and 
    \begin{equation}
        \partial_{r_{ab}} \theta_a^*(\Comp) = \sum_{c \in \mathcal{A}_a} n_{ac} \partial_{r_{ab}} r_{ac} =  \sum_{c \in \mathcal{A}_a} n_{ac} \delta_{b=c} = n_{ab} \geq 0.
    \end{equation}
    According to Proposition~\ref{prop:onpartialrab}, this implies that $\Theta_{\text{Gaussian},\sigma^2}^*(\Comp)$ is strictly increasing with respect to $\Comp$.
\end{proof}

\paragraph{Non-Lipschitz-Resilience of the Gaussian-GBT Model.}
As we have seen in Proposition~\ref{prop:gausssolution}, $\Theta_{\text{Gaussian},\sigma^2}^*(\Comp)$ depends linearly on $\Comp$. According to Proposition~\ref{prop:notresilient} below, this implies that the estimator is not Lipschitz-resilient.
The Gaussian estimation corresponds to scores  that are averaging over comparisons. 
This reflects the well-known fact that the mean is highly sensitive to manipulation in information models. 

       \begin{proposition} \label{prop:notresilient}
        If $\Dc = \R$ and $\Comp \mapsto \Theta_{\text{Gaussian},\sigma^2}^*(\Comp)$ is linear and nonzero, then the maximum likelihood estimator $\Theta_{\text{Gaussian},\sigma^2}^*(\Comp)$ is not Lipschitz-resilient.
    \end{proposition}

    \begin{proof}
        Let $\Comp$ be such that $\Theta_{\text{Gaussian},\sigma^2}^*(\Comp) \neq 0$ and $\lambda > 0$. Then, we have $\| \Theta_{\text{Gaussian},\sigma^2}^* (\lambda \Comp) \|_2 = \lambda \| \Theta_{\text{Gaussian},\sigma^2}^* \|_2$, hence $\Theta^*$ is unbounded. 
        Assume that $\Theta_{\text{Gaussian},\sigma^2}^*(\Comp)$ is Lipschitz-resilient. Then, \eqref{eq:widehattheta} applied to $\Comp' = 0$ implies that $\| \Theta_{\text{Gaussian},\sigma^2}^*(\Comp) \|_2 \leq \beta \| \Comp \|_0 \leq \beta C$ and therefore $\| \Theta_{\text{Gaussian},\sigma^2}^*(\Comp) \|_2$ is bounded. This is a contradiction, hence $\Theta_{\text{Gaussian},\sigma^2}^*(\Comp)$ is not $\beta$-Lipschitz-resilient. This is true for any $\beta > 0$, which concludes the proof.
    \end{proof}

    \subsection{$\beta$-GBT Models} \label{sec:betaBT}

    We consider symmetric $\beta$ distributions, parameterized as\footnote{The $\beta$ distributions are usually introduced over $[0,1]$ and we rescale it to $[-1,1]$. The complete family includes non-symmetric distributions parameterized by $\alpha,\beta > 0$. The symmetric case, which is an assumption on our model, is for $\alpha = \beta$. See \url{https://en.wikipedia.org/wiki/Beta_distribution} for more details.} 
    \begin{equation} \label{eq:betabt}
       f(r) = \frac{\Gamma(2\beta)}{4^{\beta} \Gamma(\beta)^2} (1-r^2)^{\beta-1} 1_{[-1,1]}(r)
    \end{equation}
    for any $\beta > 0$. The comparison domain is therefore $\Dc = [-1,1]$.  The GBT model with root law \eqref{eq:betabt} is called the $\beta$-GBT model. 
    The case $\beta = 1$ corresponds to the uniform distribution. 
    
    According to Theorems \ref{theo:monotone} and \ref{theo:resilient}, MAP estimators with prior variance $\sigma^2$ based the $\beta$-GBT model \eqref{eq:betabt} are strictly increasing with $\Comp$ and $(4 \sqrt{2} \sigma^2)$-Lipschitz-resilient.

    The moment generating function of the $\beta$ distribution $f_{\beta}$ over $[0,1]$ is given by\footnote{A derivation ca be found in \url{https://proofwiki.org/wiki/Moment_Generating_Function_of_Beta_Distribution}, applied to $\beta = \alpha$ and rescaled over $[-1,1]$.}
    $M_{\beta}(\theta) = \int_{0}^1 \mathrm{e}^{\theta r}f_{\beta}(r) \mathrm{d} r = 1 + \sum_{k \geq 1} \left( \prod_{n=0}^k \frac{\beta + n}{2\beta + n} \right) \frac{\theta^k}{k!} .$
 We then obtain the moment generating function of a $\beta$ distribution rescaled over $[-1,1]$, and therefore its cumulant-generating function
    \begin{equation} \label{eq:grostruc}
      \Phi_{\beta}(\theta) = \log \left( \frac{M_{\beta}(\theta) + M_{\beta}(- \theta)}{2}  \right)
      = 
      \log \left(  1 + \sum_{k \geq 1} \left( \prod_{n=0}^{2k} \frac{\beta + n}{2\beta + n} \right) \frac{\theta^{2k}}{(2k)!}  \right)
    \end{equation}
    The relation \eqref{eq:grostruc} admits analytic solutions for specific values of $\beta$. We illustrate this fact with two examples for $\beta = 1$ (uniform distribution) and $\beta = 2$. 

    \paragraph{Uniform-GBT Model.} If we select $\beta = 1$ in \eqref{eq:betabt} we recover the uniform distribution $f(r) = \frac{1_{[-1,1]}(r)}{2}$. In this case, we have that
 $\int_{-1}^1 \mathrm{e}^{\theta r}f(r) \mathrm{d} r 
        =  \frac{1}{2} \int_{-1}^1 \mathrm{e}^{\theta r} \mathrm{d} r 
        = \frac{\sinh \theta}{\theta}$, 
    which means that 
    \begin{equation}
        \Phi_{\text{uniform}}(\theta) = \Phi_1(\theta) = \log \left(  \frac{\sinh \theta}{\theta} \right).
    \end{equation}
    
    \paragraph{$\beta$-GBT Model with $\beta = 2$.}
    In this case, we have that 
    $ f(r) = \frac{3 }{4}(1-r^2) 1_{[-1,1]} (r)$.
    One can show\footnote{\url{https://www.wolframalpha.com/input?i=int+3\%2F4+\%281-r\%5E2\%29+e\%5E\%28r+theta\%29+dr\%2C+r\%3D-1+to+1}.} that 
    $ \int_{-1}^1 \mathrm{e}^{\theta r}f(r) \mathrm{d} r 
 =  \frac{3 }{4} \int_{-1}^1 \mathrm{e}^{\theta r} (1-r^2) \mathrm{d} r 
 =   \frac{3 (\theta \cosh \theta - \sinh \theta ))}{\theta^3},$
    which implies that 
    \begin{equation}
        \Phi_2(\theta) = \log \left( \frac{3 (\theta \cosh \theta - \sinh \theta ))}{\theta^3}\right).
    \end{equation}

\end{document}